\documentclass[journal=jctcce,manuscript=article,layout=twocolumn,]{achemso}

\usepackage{hyperref}
\usepackage[T1]{fontenc} 
\usepackage{amsmath}
\usepackage{amssymb}
\usepackage{amsfonts}
\usepackage{amsthm}
\usepackage{bm}
\usepackage{xr}
\usepackage[normalem]{ulem} 

\author{Jonas Köhler}
\affiliation{Department of Mathematics and Computer Science, Freie Universität Berlin, Arnimallee 12, 14195 Berlin, Germany}
\altaffiliation{J.K., Y.C.~and A.K.~contributed equally to this work.}

\author{Yaoyi Chen}
\affiliation{Department of Mathematics and Computer Science, Freie Universität Berlin, Arnimallee 12, 14195 Berlin, Germany}
\altaffiliation{J.K., Y.C.~and A.K.~contributed equally to this work.}

\author{Andreas Krämer}
\email{andreas.kraemer@fu-berlin.de}
\affiliation{Department of Mathematics and Computer Science, Freie Universität Berlin, Arnimallee 12, 14195 Berlin, Germany}
\altaffiliation{J.K., Y.C.~and A.K.~contributed equally to this work.}

\author{Cecilia Clementi}
\email{cecilia.clementi@fu-berlin.de}
\affiliation{Department of Physics, Freie Universität Berlin, Arnimallee 12, 14195 Berlin, Germany}
\alsoaffiliation{Center for Theoretical Biological Physics, Rice University, Houston, TX 77005, USA}
\alsoaffiliation{Department of Physics, Rice University, Houston, TX 77005, USA}
\alsoaffiliation{Department of Chemistry, Rice University, Houston, TX 77005, USA}

\author{Frank Noé}
\email{frank.noe@fu-berlin.de}
\affiliation{Microsoft Research AI4Science, Karl-Liebknecht Str. 32, 10178 Berlin, Germany}
\alsoaffiliation{Department of Mathematics and Computer Science, Freie Universität Berlin, Arnimallee 12, 14195 Berlin, Germany}
\alsoaffiliation{Department of Physics, Freie Universität Berlin, Arnimallee 12, 14195 Berlin, Germany}
\alsoaffiliation{Department of Chemistry, Rice University, Houston, TX 77005, USA}

\title{Flow-matching -- efficient coarse-graining of molecular dynamics without forces \footnote{This is the preprint of a paper published on \emph{J. Chem. Theory Comput.}~(DOI:~\href{https://doi.org/10.1021/acs.jctc.3c00016}{10.1021/acs.jctc.3c00016}) and does not contain the editing and minor changes after submission.}}


\begin{document}

\begin{tocentry}
\includegraphics{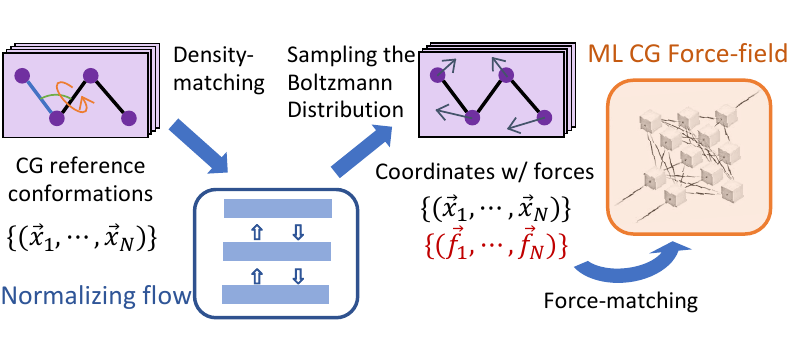}
\end{tocentry}

\begin{abstract}
Coarse-grained (CG) molecular simulations have become a standard tool to study molecular processes on time-~and length-scales inaccessible to all-atom simulations. 
Parameterizing CG force fields to match all-atom simulations has mainly relied on force-matching or relative entropy minimization, which require many samples from costly simulations with all-atom or CG resolutions, respectively.
Here we present \emph{flow-matching}, a new training method for CG force fields that combines the advantages of both methods by leveraging normalizing flows, a generative deep learning method.
Flow-matching first trains a normalizing flow to represent the CG probability density, which is equivalent to minimizing the relative entropy without requiring iterative CG simulations.
Subsequently, the flow generates samples and forces according to the learned distribution in order to train the desired CG free energy model via force-matching.
Even without requiring forces from the all-atom simulations, flow-matching outperforms classical force-matching by an order of magnitude in terms of data efficiency and produces CG models that can capture the folding and unfolding transitions of small proteins.
\end{abstract}

\section{Introduction}
Molecular dynamics (MD) simulations have become a major computational tool to study biophysical processes on molecular scales. Presently, MD simulations at all-atom resolution can reach multiple microseconds for small to medium-sized protein systems on retail hardware. By using special-purpose supercomputers~\cite{Shaw2014, Shaw2021} or combining distributed computing with Markov State Modeling~\cite{Prinz_JChemPhys2011, Husic_JACS2018} or enhanced sampling approaches, it is possible to probe millisecond-timescales and sometimes beyond \cite{Lindorff_Science2011, Plattner2017}.

Despite this progress, many biomolecular processes of interest exceed these time and length scales by orders of magnitude. Also, high-throughput simulations that would be needed, e.g., to screen protein sequences for high-affinity protein-protein interactions, cannot be efficiently done with all-atom MD.

A common approach to go to larger time-~and length-scales or high-throughput
simulations, is coarse-grained (CG) molecular dynamics~\cite{Clementi_JMolBiol2000,ClementiCOSB,Matysiak2004,Matysiak2006,das2005balancing,Saunders_AnnuRevBiophys2013,Noid_JChemPhys2013,Ingolfsson_WIREsComputMolSci2014,Kmiecik_ChemRev2016,Pak_CurrOpinStructBiol2018,Chen_JChemTheoryComput2018,Singh_IntJMolSci2019,Nuske_JChemPhys2019,wang2019machine,Wang_JChemPhys2020,husic2020coarse}.
In ``bottom-up'' coarse-graining \cite{jin2022bottom}, one defines a mapping from the all-atom representation to the CG model (e.g., by grouping sets of atoms to CG beads). The choice of mapping determines the resolution and has to suit the system as well as the scientific question, which is by itself a challenge~\cite{Noid_JChemPhys2013,Pak_CurrOpinStructBiol2018,boninsegna2018data,Wang2019}.
Given that the CG mapping is chosen, a frequently used CG principle is known as thermodynamic consistency in the coarse-graining literature and as density matching in machine learning:~the CG model should generate the same equilibrium distribution in the CG coordinates, as one would obtain from a fully converged all-atom simulation after applying the coarse-graining map to all simulation frames~\cite{Noid_JChemPhys2013}. 
In principle, the requirement of thermodynamic consistency uniquely defines the free energy function in the CG coordinates, which is also known as the \emph{potential of mean force} (PMF).
Ideally, if this thermodynamically consistent PMF were known, it could be used to compute exactly any equilibrium property expressible as an ensemble average over the CG coordinates.
Note that this definition does not guarantee that the CG model reproduces \emph{all} thermodynamic observables, counterexamples being heat capacity, pressure, and entropy~\cite{wagner2016representability, dunn2016van, jin2019understanding, dannenhoffer2019compatible, doi:10.1063/1.5125246}.
However, the PMF by definition involves high dimensional integrals that cannot be estimated for nontrivial systems in practice.
A pivotal challenge is to find a good approximation for the PMF with tractable functional forms to serve as the CG potential~\cite{Noid_JChemPhys2013}.

\begin{figure}[htb!]
    \centering
    \includegraphics[width=0.9\linewidth]{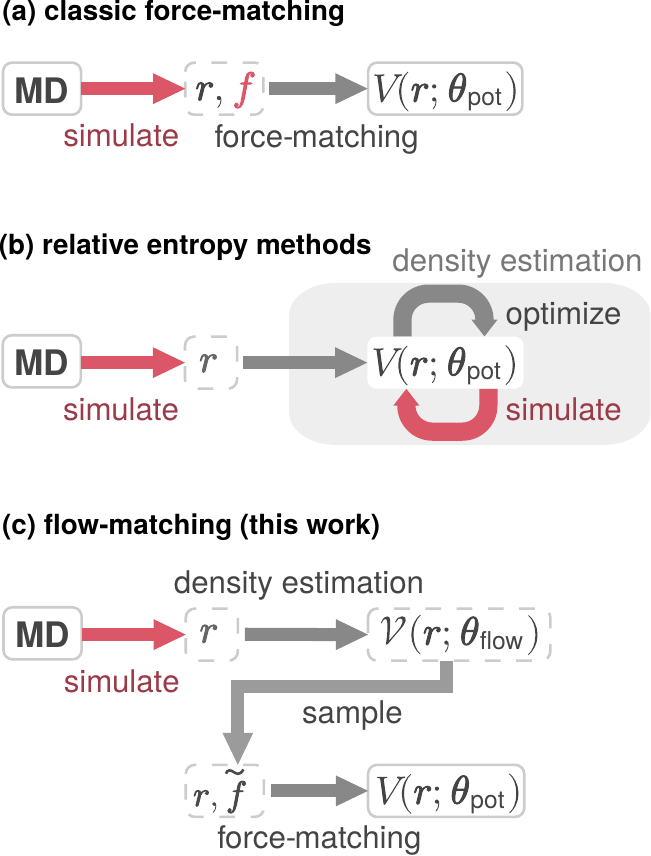}
    \caption{
    Overview of the flow-matching method. \textbf{a)} Classical force-matching:~the parameters $\bm\theta_{\mathrm{pot}}$ of a CG potential $V(\cdot; \bm\theta_{\mathrm{pot}})$ are optimized to minimize the mean-squared error of model forces with respect to projected atomistic forces $\bm{f}$ on the training configurations $\bm{r}$;  \textbf{b)} Relative entropy methods:~simulations are performed with the CG potential to produce samples and enable evaluating (and minimizing) the relative entropy;
    \textbf{c)} Present method:~the parameters $\bm\theta_{\mathrm{flow}}$ of a normalizing flow are first optimized to match the CG density from the ground-truth samples $\bm{r}$. This defines the flow-based potential $\mathcal{V}(\cdot; \bm\theta_{\mathrm{flow}})$. The samples and forces from the flow are used to train a CG potential $V(\cdot; \bm\theta_{\mathrm{pot}})$ via force-matching.
    Slow/inaccurate sampling steps are highlighted in red.
    }
    \label{fig:method-overview}
\end{figure}

Among the techniques for such bottom-up modeling~\cite{Reith2003,Izvekov_JPhysChemB2005,Noid_JChemPhys2008,Shell_JChemPhys2008, Noid_JChemPhys2013, Kmiecik_ChemRev2016}, two methods have been explicitly developed to approach thermodynamic consistency:~variational force-matching (also known as multi-scale coarse-graining)~\cite{Izvekov_JPhysChemB2005,Noid_JChemPhys2008} and relative entropy minimization~\cite{Shell_JChemPhys2008}.
Force-matching (Fig.~\ref{fig:method-overview}a) is straightforward to implement but requires the forces on the CG particles mapped from all-atom sampling. Because these instantaneous forces depend on all degrees of freedom, they provide a very noisy signal that makes training the CG force field data inefficient.
This approach has been connected with the blooming field of machine-learned potentials and led to several successes~\cite{wang2019machine,husic2020coarse,Wang_JChemPhys2020}.
Relative entropy minimization (Fig.~\ref{fig:method-overview}b), the Inverse Monte-Carlo method \cite{PhysRevE.52.3730}, as well as Iterative Boltzmann Inversion~\cite{Reith2003}, do not require forces to be recorded and are more data-efficient, but require the CG model to be re-simulated during the iterative training procedure, which can be extremely costly and even lead to failure in convergence.
Ref.~\citenum{Pak2019} developed a hybrid approach combining force-matching and relative entropy methods in order to parameterize CG models where not all particles have force information available.

This manuscript presents a third alternative---the \emph{flow-matching} method, which is shown to be more efficient. Our approach combines aspects of force-matching and relative entropy minimization with deep generative modeling. The centerpiece of this novel method is a \emph{normalizing flow}~\cite{tabak2010density, rezende2015variational, papamakarios2019normalizing}, a generative neural network that can approximate arbitrary probability distributions by transforming them into simple, easy-to-sample prior distributions. Once trained, normalizing flows can generate uncorrelated samples and compute normalized probability densities, energies, and forces, which makes them an exciting emerging tool for physical applications~\cite{noe2019boltzmann, Gabrie2022,Li2020,Nicoli2020,Liu2021,Ding2021,wirnsberger2020targeted,Ding2021}. For example, Boltzmann generators \cite{noe2019boltzmann} use flows that are trained on MD data and energies as one-shot importance samplers for molecular equilibrium distributions.
Other types of generative neural networks have also been used for back-mapping of CG structures~\cite{DBLP:conf/icml/WangXCMSW0G22, Stieffenhofer2020}.

Flow-matching applies normalizing flows to the coarse-graining problem. Like force-matching and relative entropy minimization, it starts from CG samples in equilibrium, which are usually created by mapping snapshots from an all-atom simulation to the CG space. In order to find a thermodynamically consistent CG potential, the method proceeds in two steps (Fig.~\ref{fig:method-overview}c). First, a normalizing flow is trained by density matching, such that it learns to sample directly from the target ensemble. Second, the CG equilibrium distribution that the flow has learned is taught to a CG force field by force-matching to coordinate-force pairs generated by the flow.

While this stepwise approach leans on the same underlying principles as previous coarse-graining methods, it avoids their key limitations. In contrast to force-matching (Fig.~\ref{fig:method-overview}a,c), it does not rely on atomistic reference forces.
Although saving forces during the simulation is in principle straightforward to do, in most cases of already existing long simulations, forces have not been stored and are often non-trivial to recompute \textit{a posteriori}.
To bypass the need for force data, an alternative method has been previously proposed as the generalized Yvon-Born-Green theory~\cite{Mullinax2009}, which determines a CG force field (usually as a sum of basis functions) directly according to structural correlations.
However, it is not clear whether this can be generalized to CG force fields based on neural networks. 

Additionally, the flow can generate an indefinite number of ``synthetic'' configurations and forces, which do not carry noise from the atomistic environment. %
In contrast to relative entropy minimization~\cite{Shell_JChemPhys2008} and iterative Boltzmann inversion~\cite{Reith2003}, flow-matching does not require repeated re-simulation of the CG model during training, as the flow can generate independent samples that represent the thermodynamic equilibrium (Fig.~\ref{fig:method-overview}b,c). In practice, by removing the need for costly simulations during training, flow-matching makes coarse-graining by density estimation/relative entropy methods feasible for molecules with rare events, such as biomolecules. In contrast to force-matching, density estimation does not suffer from the noise problem due to the omitted degrees of freedom, and consequently, flow-matching is significantly more data-efficient.

Using the flow only as an intermediate offers complete freedom in choosing the functional form of the final CG force field. In particular, the candidate potential can incorporate the desired physical symmetries and asymptotics \cite{wang2019machine,Wang_JChemPhys2020} as well as share parameters across chemical space \cite{husic2020coarse}. Conversely, directly using a normalizing flow as the CG force field would not be a good idea, because transferable properties cannot be easily incorporated into invertible \cite{tabak2010density, papamakarios2019normalizing} or at least statistically reversible \cite{wu2020snf} neural networks, which are required by the flows. For example, transferability across molecular systems of different sizes and topologies requires parameter sharing and a transformation of random variables of different dimensionality---features not yet supported by existing normalizing flows.
To this end, flow-matching combines the advantages of normalizing flows and energy-based models in a novel way. 
Flow-matching \textit{per se} does not enable transferability. However, it helps towards this goal by allowing the training of neural network force fields in a data-efficient way, thus significantly reducing the burden of generating extensive training data.

As a proof of concept, we apply the method to the coarse-graining of small protein systems. We show that accurate CG models can be fit to equilibrium trajectories without using atomistic forces or intermediate simulations. Even when forces are available, we find that flow-matching is much more data-efficient than force-matching and yields surprisingly accurate force fields on small data sets.

\section{Coarse-graining with Flow-matching}

\begin{figure*}[htb!]
    \centering
    \includegraphics[width=1.0\linewidth]{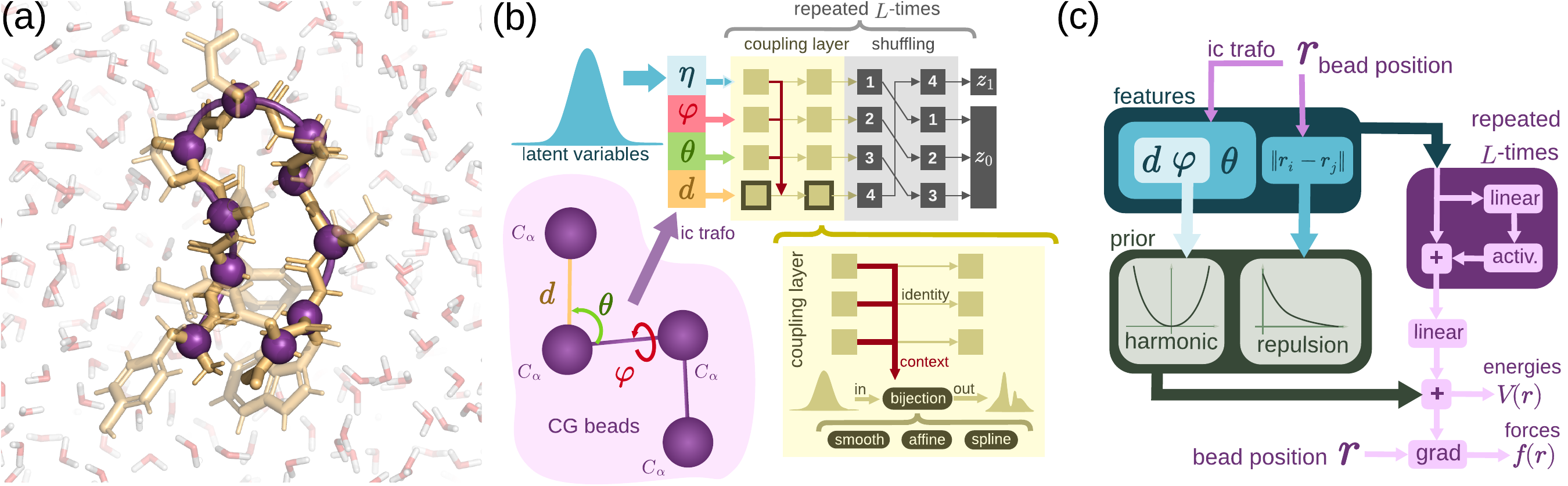}
    \caption{
    \textbf{(a)} Chignolin in explicit solvent. The magenta spheres show the CG beads at C$_\alpha$ resolution.
    \textbf{(b)} The normalizing flow architecture used in this work to model $\mathcal{V}(\cdot;\bm\theta_{\mathrm{flow}})$. After transforming the CG beads into an internal coordinate (IC) representation made from bonds ($d$), angles ($\theta$), and dihedral torsions ($\varphi$) a trainable stack of coupling layers transform them into uniform noise. See Fig.~S1 for a more detailed illustration of the flow architecture.
    \textbf{(c)} The modified \textit{CG-Net} architecture used in this work to model $\mathcal{V}(\cdot;\bm\theta_{\mathrm{pot}})$. ``grad'' stands for computing the gradient using automatic differentiation.}
    \label{fig:method-architecture}
\end{figure*}

\subsection{Coarse-graining with thermodynamic consistency}
We consider a molecular system with atomic coordinates $\bm{R} \in \mathbb{R}^{3 N}$ in thermodynamic equilibrium following an equilibrium distribution 
\begin{align}
    \mu(\bm{R}) \propto \exp(-u(\bm{R}))
\end{align}
where $u$ is the reduced potential energy of the system, whose exact form depends on the choice of the ensemble, e.g. ${u(\bm R)=U(\bm R) / kT}$ for the canonical ensemble with potential energy $U(\bm R)$, temperature $T$ and Boltzmann constant $k$.

Coarse-graining considers a mapping $\bm\Xi \colon \mathbb{R}^{3 N} \rightarrow \mathbb{R}^{3n}$ that projects fine-grained states $\bm R$ onto a lower-dimensional representation $\bm r$. In the present work, we only consider linear and orthogonal maps, $\bm r=\bm \Xi R$. For non-orthogonal or even nonlinear maps, the subsequent mathematical treatment must be generalized \cite{Ciccotti_CommunPureApplMath2008, Kalligiannaki2015}.
As an example, the conformational dynamics of a protein with $N$ atoms can be projected onto a chosen set of beads by only considering the $C_{\alpha}$-atoms in the backbone (Fig.~\ref{fig:method-architecture}a). Coarse-graining with thermodynamic consistency aims at parameterizing a CG model which yields the same density over the CG coordinates as the marginal distribution from the original system, i.e.,
\begin{align}
    \nu(\bm r) = \int d\bm R ~ \mu(\bm R)  \cdot \delta_{[\bm\Xi \bm R = \bm r]}(\bm R).
    \label{eq:PMF-def}
\end{align}
The CG model is often defined by a CG potential $V(\cdot; \bm\theta_{\mathrm{pot}})$ with parameters:~$\nu(\cdot; \bm\theta_{\mathrm{pot}}) \propto \exp(-V(\cdot; \bm\theta_{\mathrm{pot}}))$.
Two conventional parameterization approaches will be introduced below.
It is important to stress that designing a CG force-field by trying to optimize \textit{thermodynamic} consistency does not imply that also the \textit{dynamical} properties are well approximated~\cite{Davtyan2016,Nuske_JChemPhys2019}. 

\subsection{Variational force-matching}
One option is to optimize a candidate potential $V(\cdot;\bm\theta_{\mathrm{pot}})$ with the force information from the ground-truth potential $u$ (Fig~\ref{fig:method-overview}a).
Given a set of fine-grained samples (e.g., MD trajectory) ${\mathcal{D} = \left(\bm R_{1}, \ldots, \bm R_{T}\right)}$ with corresponding forces ${\bm f(\bm R) = - \nabla u(\bm R)}$, 
it is shown that the thermodynamically consistent CG potential (Eq.~\eqref{eq:PMF-def}) can be approximated by the potential
minimizing the \emph{variational force-matching loss} \cite{Noid_JChemPhys2008}
\begin{align}
    \mathcal{L}(\bm\theta_{\mathrm{pot}}) = \mathbb{E}_{\bm R, \bm f \sim \mathcal{D}}\left[\left\| {\bm \Xi}_f \bm f + \nabla_{\bm \Xi \bm R} V(\bm \Xi \bm R; \bm\theta_{\mathrm{pot}}) \right\|_{2}^{2} \right], \label{eq:force-matching-objective}
\end{align}
in which ${\bm \Xi}_f$ is a force mapping operator dependent on map ${\bm \Xi}$. 
When infinite samples $\mathcal{D}$ and all functional forms for $V$ are available, the minimization of the loss (Eq.~\eqref{eq:force-matching-objective}) yields exactly the thermodynamically consistent potential defined by Eq.~\eqref{eq:PMF-def}. Even with finite samples and restrictions on the $V(\cdot;\bm\theta_{\mathrm{pot}})$, the result from the loss minimization still provides a variational approximation in practice.
Because of their enhanced expressiveness, neural networks with physical inductive biases have been shown to be a useful model class for the parameterization of $V(\cdot;\bm\theta_{\mathrm{pot}})$~\cite{wang2019machine, husic2020coarse}.

\subsection{Density estimation / relative entropy method}
Force-matching requires the mapped CG forces to be saved during fine-grained sampling, which is not common practice. Alternatively, one can directly learn a CG model via density estimation on the observed conformational space. Density estimation aims at minimizing the following objective
\begin{align}
    \mathcal{L}(\bm\theta_{\mathrm{pot}}) = \mathbb{E}_{\bm R \sim \mathcal{D}}\left[-\log \nu(\bm \Xi \bm R; \bm\theta_{\mathrm{pot}})\right]. \label{eq:relative-entropy-objective}
\end{align}
The minimum can be interpreted as the maximum-likelihood solution of an energy-based model trained on the projected samples ${\bm \Xi\mathcal{D} = \left(\bm \Xi \bm R_{1}, \ldots \bm \Xi \bm R_{T}\right)}$.
This approach can be related to the relative entropy method in molecular simulation \cite{Shell_JChemPhys2008} and is used for training an energy-based model in the field of machine learning~\cite{lecun2007}. Unfortunately, computing the gradients of Eq.~\eqref{eq:relative-entropy-objective} with respect to parameters generates a sampling problem. Computing the gradient contribution of the normalizing constant involves sampling from the model density $\nu$, which means that the CG model needs to be periodically re-sampled during training (Fig~\ref{fig:method-overview}b).

\subsection{Flow-based density estimation}
We can avoid the sampling problem of Eq.~\eqref{eq:relative-entropy-objective} by using the density $\nu(\cdot; \bm\theta_{\mathrm{flow}})$ corresponding to a model that can be efficiently sampled, such as normalizing flows  \cite{tabak2010density, rezende2015variational, papamakarios2019normalizing}. Flows are invertible neural networks $\Phi(\cdot; \bm\theta_{\mathrm{flow}}) \colon \mathbb{R}^{n} \rightarrow \mathbb{R}^{n}$ that transform an easy-to-sample reference distribution $q(\bm z)$, e.g., a Gaussian or uniform density, into our target density. If we  sample $\bm z \sim q(\bm z)$ and transform it into $\bm r = \Phi(\bm z; \bm\theta_{\mathrm{flow}})$ the resulting density is given by
\begin{align}
    p(\bm r; \bm\theta_{\mathrm{flow}}) = q\left(\Phi^{-1}(\bm r; \bm\theta_{\mathrm{flow}})\right) \cdot \left|J_{\Phi^{-1}}(\bm r; \bm\theta_{\mathrm{flow}})\right|. \label{eq:flow-density}
\end{align}
Inserting Eq.~\eqref{eq:flow-density} into Eq.~\eqref{eq:relative-entropy-objective} we get an efficient training objective. After training, the energy of the normalizing flow
\begin{align}
    \mathcal{V}(\bm r; \bm\theta_{\mathrm{flow}}) = - \log p(\bm r; \bm\theta_{\mathrm{flow}})
\end{align}
approximates the CG PMF.

\subsection{Variational density estimation}
Direct density estimation with flow models suffers from the fact that the flow architecture is constrained in order to represent an invertible function, which compromises their representative power and training dynamics.
As a solution, we consider relaxing the bijectivity constraint by introducing $k$ additional variables and sampling a joint state $\bm z= (\bm z_0, \bm z_1) \in \mathbb{R}^{n + k}$ from a joint (Gaussian/uniform) reference density $q(\bm z_0, \bm z_1)$~(Fig.~\ref{fig:method-architecture}b). Now we define $\Phi$ as an invertible coordinate transformation (e.g., a flow model) over those joint $n+k$ degrees of freedom. Similarly as before, we get the output density $p(\bm r, \bm \eta; \bm\theta_{\mathrm{flow}})$ of a transformed pair ${(\bm r, \bm \eta) = \Phi(\bm z_0, \bm z_1; \bm\theta_{\mathrm{flow}})}$. 
The marginal density over $\bm r$ of this augmented model cannot be computed efficiently. However, we can still optimize a variational bound of the likelihood: we first define a joint density ${\nu(\bm r, \bm \eta) = \nu(\bm r) \cdot \bm \tilde\nu(\bm\eta|\bm r)}$ by introducing a Gaussian conditional density $\tilde\nu(\bm \eta|\bm r)$ and then minimize
\begin{align}
    \mathcal{L}(\bm\theta_{\mathrm{flow}}) &= \mathbb{E}_{\bm R \sim \mathcal{D}, \bm \eta \sim \tilde\nu(\bm \eta|\bm r)}\left[-\log p(\bm \Xi \bm R, \bm \eta; \bm\theta_{\mathrm{flow}})\right] \label{eq:variational-relative-entropy-objective} \\
    &\geq \mathbb{E}_{\bm R \sim \mathcal{D}}\left[-\log p(\bm \Xi \bm R; \bm\theta_{\mathrm{flow}})\right]. \nonumber
\end{align}
As shown in \cite{huang2020augmented, chen2020vflow}, normalizing flows with additional noise dimensions can alleviate limitations of invertible neural networks to transform a simple, unimodal, prior density to a complex, multimodal target density \cite{cornish2020relaxing, wu2020snf, brofos2021manifold}. While the extra dimensions do not allow us to directly compute the density $p(\bm r; \bm\theta_{\mathrm{flow}}),$ and thus $\mathcal{V}(\bm r, \bm\theta_{\mathrm{flow}})$ as well as the corresponding forces, we can still compute a joint energy model over CG coordinates and latent variables 
\begin{align}
    \mathcal{V}(\bm r, \bm \eta;\bm\theta_{\mathrm{flow}}) = - \log p(\bm r, \bm \eta; \bm\theta_{\mathrm{flow}}), 
\end{align}
which can be used to train an arbitrary model of the CG potential as follows.

\subsection{Teacher-student force-matching}
Our idea is to teach the information about the distribution of the CG coordinates $\bm r$ contained in a trained latent-variable model $\mathcal{V}(\bm r, \bm \eta;\bm\theta_{\mathrm{flow}})$ to a ``student'' CG potential $V(\bm r; \bm\theta_{\mathrm{pot}})$ that does not suffer from the architectural constraints of flows (Fig.~\ref{fig:method-architecture}c).
We first draw samples $(\bm r, \bm \eta)$ from our flow model and compute instantaneous forces over CG coordinates $\bm r$:
\begin{align}
    \tilde {\bm f}(\bm r, \bm \eta; \bm\theta_{\mathrm{flow}}) = -\nabla_{\bm r} \mathcal{V}(\bm r, \bm \eta; \bm\theta_{\mathrm{flow}}) \label{eq:instantaneous-forces}.
\end{align}
Any given $\bm r$ may correspond to different $\tilde {\bm f}$, but on average they give rise to the unbiased mean force: 
\begin{align}
    \bm f(\bm r; \bm\theta_{\mathrm{flow}}) = \mathbb{E}_{\bm \eta \sim p(\bm \nu | \bm r; \bm\theta_{\mathrm{flow}})}\left[ \tilde {\bm f}(\bm r, \bm \eta; \bm\theta_{\mathrm{flow}}) \right].
    \label{eq:flow-mean-force}
\end{align}
This relation allows us to efficiently train an unconstrained $V(\bm r; \bm\theta_{\mathrm{pot}})$ via the variational force-matching objective
\begin{align}
    \mathcal{L}(\bm\theta_{\mathrm{pot}})
    &= \mathbb{E}_{(\bm r, \bm \eta) \sim p(\bm\theta_{\mathrm{flow}})}\left[ \left\| \tilde {\bm f}(\bm r, \bm \eta; \bm\theta_{\mathrm{flow}}) + \nabla_{\bm r} V(\bm r; \bm\theta_{\mathrm{pot}})  \right\|^2_2 \right]. \label{eq:variational-force-matching} %
\end{align}
As shown in the supplementary information (SI), the gradients of Eq.~\eqref{eq:variational-force-matching} with respect to the $\bm\theta_{\mathrm{pot}}$ provides an unbiased estimator that does not depend on $\bm\theta_{\mathrm{flow}}$.
The proposed approach resembles conventional  for coarse-graining, but with the difference that it averages over fewer degrees $\bm\eta$ rather than a larger amount of (mainly solvent) degrees of freedom.

As will be shown in the Results, the student model can mitigate flaws in the flow models, namely samples that deviate from physics laws (e.g., containing steric clashes) and the ruggedness of the CG free energy surface.
The student model is also regularized to entail a more robust CG potential than the direct force output of the flow for molecular dynamics simulation.
In addition, the flexibility in choosing the functional form of the CG free energy allows built-in symmetries such as roto-translational energy invariance \cite{wang2019machine} and parameter sharing for obtaining a transferable force field \cite{husic2020coarse}.

\section{Results}
We now employ the flow-matching method to obtain CG molecular models of small proteins. 
To this end, we train flows on the CG coordinate samples extracted from all-atom simulation trajectories. Trained flow models can generate CG coordinates and accompanying forces, which in turn are used to train a neural CG potential via force-matching. 
For demonstration purposes, this work uses an improved version of the CGnet architecture~\cite{wang2019machine} to represent the CG force field (Fig.~\ref{fig:method-architecture}c; see also Methods in the SI). Therefore, these secondary CG models will be denoted as ``Flow-CGnets''.

\subsection{Flow-matching learns accurate CG force fields}
\begin{figure}[htb!]
    \centering
    \includegraphics[width=\linewidth]{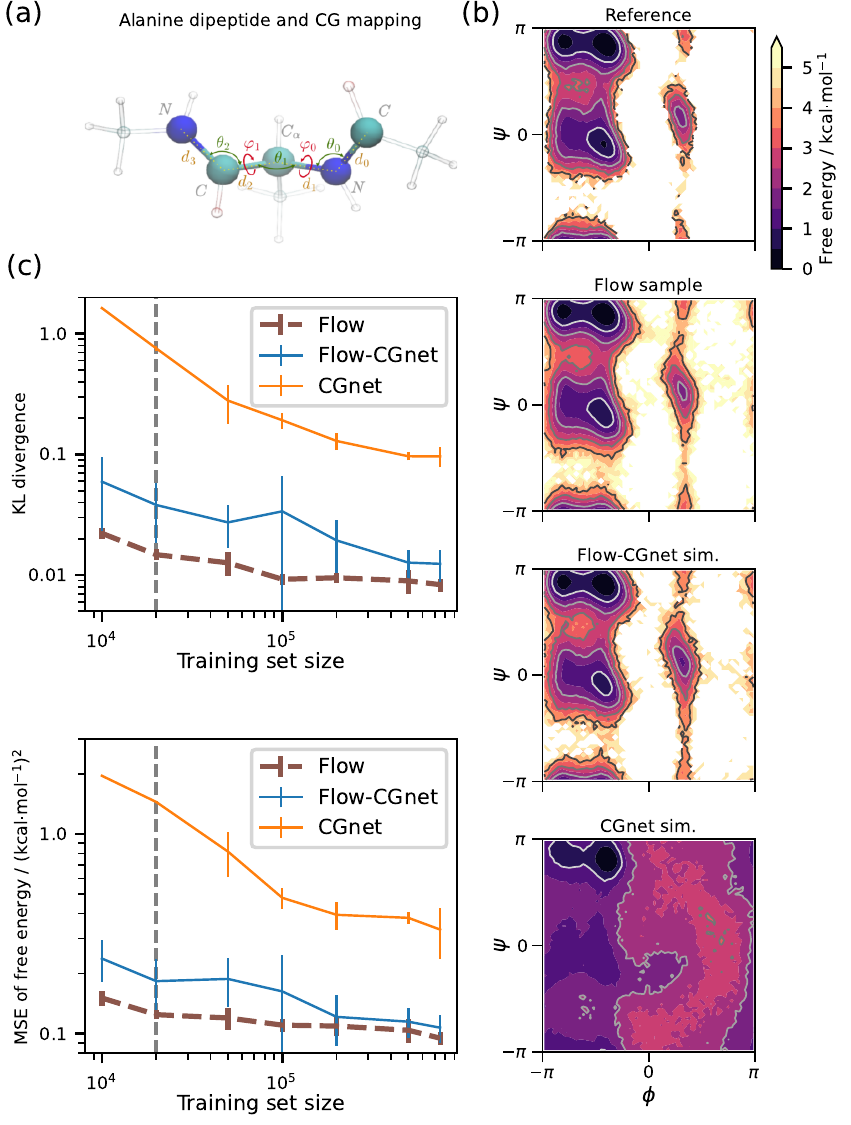}
    \caption{
    Application of flow-matching on capped alanine.
     \textbf{a)} The CG mapping used for the flow and CGnets, $\varphi_0,\:\varphi_1$ represents main chain torsion angles $\phi$ and $\psi$, respectively.
     \textbf{b)} Free energy profile of capped alanine projected on the $\phi/\psi$ plane (Ramachandran plot) for the all-atom ground truth from MD simulation (ground truth), for the flow model, for the Flow-CGnet and for original CGnet model (baseline). The latter three were trained against only 20,000 data points from the reference data (vertical grey dashed lines in \textbf{(c)}).
     \textbf{c)} Model accuracy as a function of training set size for capped alanine. Shown metrics are estimated KL divergence and MSE between discrete free energies on the $\phi/\psi$ plane. Brown dashed curves correspond to the flow after MLE training, while solid lines show values for the \textit{CGNet}s trained on either the flow sample (blue) or the all-atom ground truth sample (orange).
     }
    \label{fig:ala2-new-figure}
\end{figure}

As a first example, we consider capped alanine, also known as alanine dipeptide, to demonstrate that flow-matching can learn accurate CG force fields and achieve much higher statistical efficiency than force-matching.
As in previous work~\cite{wang2019machine, husic2020coarse}, the CG mapping is defined as slicing out the coordinates of five backbone carbons and nitrogens (Fig.~\ref{fig:ala2-new-figure}a).

We quantify the accuracy of different methods based on equilibrium statistics from either direct sampling (for flows) or long simulation trajectories (for CGnets).
We focus on the joint distributions of the $\phi-$ and $\psi-$dihedral angles along the backbone (i.e., Ramachandran plot, Fig.~\ref{fig:ala2-new-figure}b), which are the main degrees of freedom for this system~\cite{Tobias_JPhysChem1992}.
The ground truth for comparison comes from all-atom MD simulation (2 microseconds in total, see Methods in the SI).
As for baseline, we use CGnets trained with classical force-matching~\cite{Noid_JChemPhys2008, wang2019machine} employing forces stored during all-atom simulations.
As illustrated by Fig.~\ref{fig:ala2-new-figure}b, the flow and Flow-CGnet can recover the reference distribution to a very good approximation when only 20,000 reference all-atom conformations are used.
In contrast, a normal CGnet cannot effectively model the dihedral free energy in this low data regime, even with the additionally available force information:~The free energy minima are more or less located according
to the ground truth (representative conformations from all-atom and two CGnet models illustrated and compared in SI Fig.~S6), but the dihedral distribution smears
over the whole space. When increasing the amount of training data, also CGnet trained with force-matching can well approximate the free energy landscape (as reported in Ref.~\citenum{wang2019machine}, where $8 \times 10^5$ configurations and forces were used), but never reaches the flow-matching accuracy for the available dataset (Fig.~\ref{fig:ala2-new-figure}c). 
This comparison displays the advantage of the flow-matching method, which infers the boundary of free energy basins as well as relative weighting between different metastable states better than force-matching, especially for regions rarely covered by the training data, e.g., at transition states.

\subsection{Flow-matching is more data efficient than force-matching}
The better accuracy of Flow-CGnet models can be attributed to higher statistical efficiency.
For illustration, we measure the effects of the training set size on %
the KL divergence and mean square error of torsional free energy, which are
computed on a discrete histogram against the validation data distribution~\cite{husic2020coarse}.
Concretely, we perform training with a varying number of
samples in the training set for both flow-matching and baseline force-matching. 
Detailed training setup can be found in SI.

It can be observed that the direct samples from the flow model ranks first regarding both criteria (Fig.~\ref{fig:ala2-new-figure}c), which renders the knowledge transfer to a student Flow-CGnet model to be ``lossy''.
Nevertheless, the secondary model provides a potential that is not only faster to evaluate, but also numerically more stable for CG molecular dynamics.
Despite that the flow model automatically provides a differentiable energy function, it is not fully accurate in regions with low Boltzmann probabilities:~A simulation with flow potential often visits spurious states outside of the distribution and sometimes experiences numerical blow-ups on the boundary of training data distribution.
This issue is solved by our two-stage training strategy, in which the CGnet can incorporate an additive, physics-inspired term (i.e., the prior energy) to set simulation-friendly free energy barriers and rule out outlier conformations~\cite{wang2019machine}.
Flow samples with unrealistically high force magnitude or located in unrealistic conformational regions can be filtered or reweighted before feeding to the CGnet training (see Methods in the SI).
The remaining samples mostly lie in the high-probability region, thus bringing informative forces for force-matching training.
As a result, the Flow-CGnet also benefits from the flow's efficiency: it achieves an equivalent performance of CGnet at full data set size even with the smallest tested input data amount (Fig.~\ref{fig:ala2-new-figure}c).

In Suppl.~Fig.~S2 we show a similar analysis of the data efficiency of CGnet, the flow and Flow-CGnet for the miniprotein chignolin. As expected, the situation is even more extreme than for alanine dipeptide:~The Flow-CGnet trained on only $2 \times 10^4$ data points is on par with the CGnet trained on all available $1.4 \times 10^6$ data points, resulting in a $70\times$ data efficiency, which is expected to further increase for larger systems.

How can the greater data efficiency of Flow-CGnet compared to force-matching be explained?
While the accuracy of the flow to approximate the Boltzmann distribution depends on the number of conformations used to train it, it achieves a very good approximation with relatively few observed conformations compared to force-matching (Fig.~S3). Although its free energy surface is not necessarily well-behaved in all local details,
the flow can generate abundant samples and forces from the learned distribution, thereby cheaply reducing the error of the trained Flow-CGnet to a similar level as the intrinsic approximation error of the flow (Fig.~S3). 
Additionally, the augmentation channels in the flow model are much fewer in number and have simpler distribution than the internal degrees of freedom in the all-atom system, and therefore the flow's sample forces have much less noise than instantaneous forces stored in all-atom simulations, and better represent the CG mean force.
In this sense, when a proper sample filtering scheme and regularizations on the CGnet models are adopted, the flow can become superior to a limited set of all-atom data in terms of the number of samples as well as the signal-to-noise ratio of forces it feeds to the secondary CGnet.
The performance in this test case suggests Flow-CGnets may extend the application of neural CG potentials to more complex macromolecular systems, where usually only a limited amount of conformations and no forces are available.

\subsection{Flow-matching of fast-folding proteins}

\begin{figure*}[htb!]
    \centering
    \includegraphics[width=17.78cm]{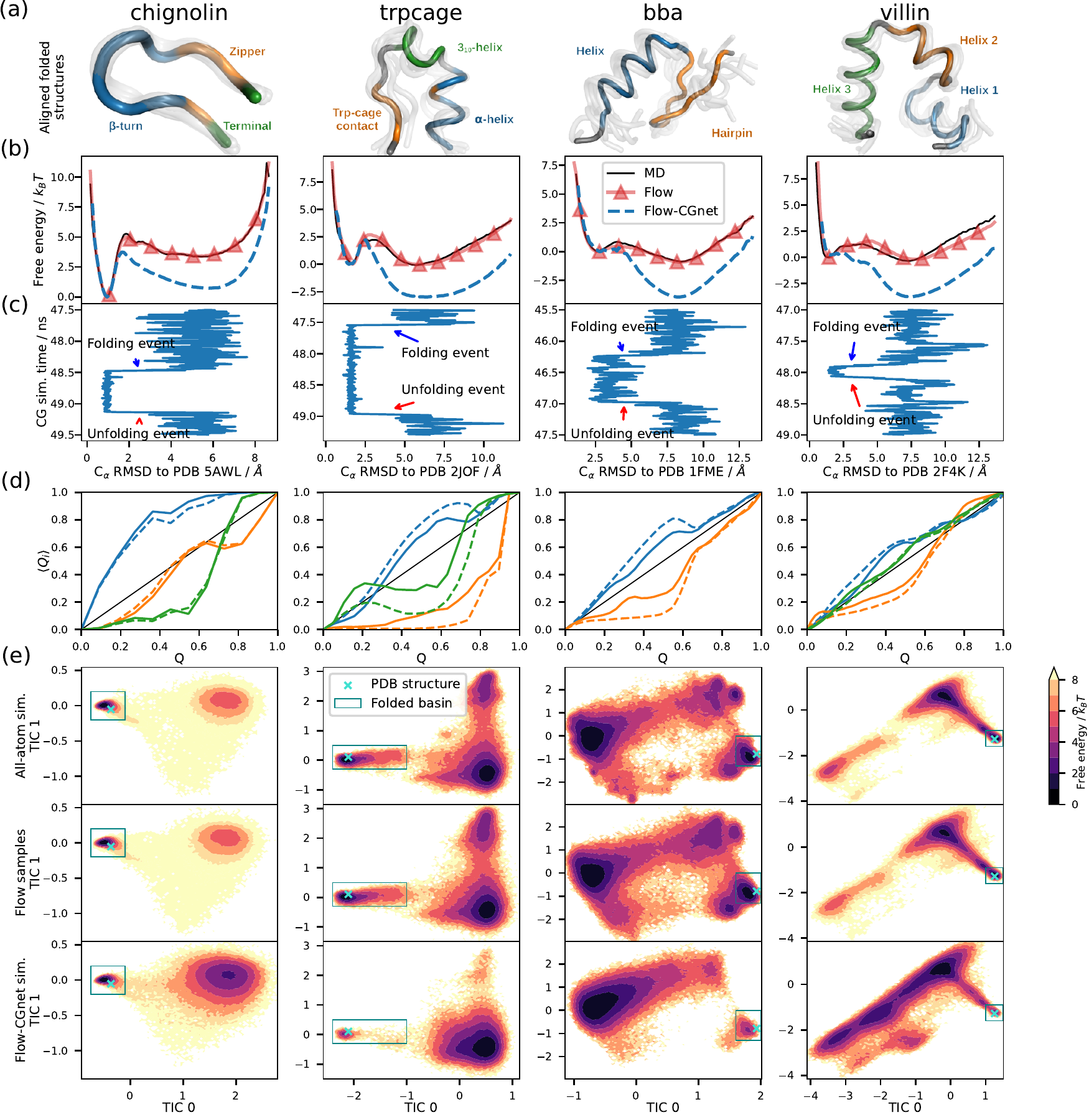}
    \caption{
    Flow-matching results for four fast-folding proteins. From top to bottom:
    \textbf{a)} 10 exemplary folded samples from CG simulation (shown in half-transparent gray color) superposed on the experimental structure. The colored segments correspond to important elements in the folding process;
    \textbf{b)} free energy curve over RMSD for the MD, flow, and Flow-CGnet samples with PDB structure as the reference;
    \textbf{c)} RMSD time series excerpt from CG simulation showing folding and unfolding events;
    \textbf{d)} Average fractions of native contacts $\langle Q_i \rangle$ in different segments of the protein formed at each stage of the folding process (identified by the fraction of all native contacts formed, $Q$). The segments are determined mainly according to secondary structures and are highlighted with the corresponding color in subfigure~(a);
    \textbf{e)} free energy landscapes of all-atom MD, flow and Flow-CGnet model over TICs, turquoise crosses and teal rectangles denoting experimental structures and folded state according to MD trajectories, respectively.
    }
    \label{fig:fastfolders}
\end{figure*}

The flow-matching method is applied to molecular trajectories of four small proteins from Ref.~\citenum{Lindorff_Science2011}, namely chignolin, tryptophan cage (trpcage), the $\alpha$/$\beta$ protein BBA (bba), and the villin headpiece (villin) that consist of 10, 20, 28, and 35 amino acids, respectively (see Ref.~\citenum{Lindorff_Science2011} for simulation details).
These small proteins can be modeled by a flow that operates fully in internal coordinates.
As for other fast folding proteins in Ref.~\citenum{Lindorff_Science2011}, some 
only have a marginally stable state that closely resembles the PDB structure throughout the all-atom trajectories, e.g., BBL; for some fast folders, we can acquire reasonable good flow models, but the folded state cannot be stabilized by the subsequent Flow-CGnet models, e.g., wwdomain and homeodomain; for the rest, the internal-coordinate-based flow model cannot effectively capture the full free energy surface (see Discussion section on scalability).
Each trajectory corresponds to at least 100 $\mu$s of all-atom MD. Note that the trajectories do not contain atomistic forces, so force-matching is not an option for parametrizing a CG force field based on these data. Relative entropy minimization is difficult because it would require iteratively re-sampling the CG model during training, introducing excessive computational cost.

All four proteins are CG using one bead per residue placed upon the C$_\alpha$ (see Fig.~\ref{fig:method-architecture}a). First, normalizing flows are trained for each protein using likelihood maximization on the C$_\alpha$ coordinates. Second, synthetic position/force pairs are generated by the flow, of which the outliers are filtered and reweighed according to the extent they exceed the force magnitude boundary and violate the minimum pairwise distances, respectively. Last, the protein-specific CGnets are optimized via force-matching on the processed flow samples.
The final CGnets are simulated using Langevin dynamics with parallel tempering to produce equilibrium samples from the CG model. The trajectories from the replica at the same temperature as the all-atom simulation are used for the analyses below.
In order to show folding and unfolding events occur without enhanced sampling strategies, we also performed pure Langevin dynamics simulations with learned Flow-CGnet models.
The details on the procedure of training and simulation as well as hyperparameter choices can be found in Methods in the SI.

\subsection{Flow-CGnets recover native structures}
Figure~\ref{fig:fastfolders} compares protein folding between the atomistic and CG simulations. All CG models recover the folded PDB structures up to $2.5$\,\AA{} RMSD, which is of similar quality as the reference all-atom simulations. Figure~\ref{fig:fastfolders}a shows representative structures from the CG simulations superposed with the experimental crystal structures, and Fig.~S7 provides a detailed comparison between the CG and atomistic structural ensembles corresponding to the different minima in the free energy landscape of the four proteins studied, demonstrating excellent agreement.
The free energy plots over the RMSD (Fig.~\ref{fig:fastfolders}b) indicate that the CG conformational distribution matches the projected all-atom trajectory for the folded basin:~The free energy valleys with the lowest RMSD values are centered around almost the same RMSD value and have nearly indistinguishable widths between the CG and MD densities, which indicates that all CG models accurately represent the flexibility of their respective folded states.

\subsection{Flow-CGnets match the folding thermodynamics qualitatively}
Moving into the unfolded region (RMSD $\geq$ 5\AA{} in Fig.~\ref{fig:fastfolders}b), the match between atomistic and CG free energies deteriorates. While all CG models exhibit the characteristic folding free energy barrier, the height of this barrier and the folded/unfolded ratio differ between the MD and CG data. Generally, the folded states are less stable in the CG model. While the flow differs by less than $\approx1\,kT$ from the all-atom result, the Flow-CGnet underestimates the folding free energy by up to $3\,kT$.

Nevertheless, frequent transitions between folded and unfolded configurations were observed in 50~ns simulation runs without parallel tempering (Fig.~\ref{fig:fastfolders}c). This assures that the models still keep the two states kinetically connected.

\subsection{Flow-CGnets reproduce the folding mechanisms}
Figure~\ref{fig:fastfolders}d illustrates the sequence of formation of the protein structure elements during folding, for the all-atom model and the corresponding Flow-CGnet model of the four proteins studied. The average fractions of native
contacts $\langle Q_i \rangle$ formed in different segments of the protein along the folding process \cite{Clementi2003} is reported and shows that the order of formation of the different secondary structure elements is recovered by Flow-CGnet to a good approximation.

\subsection{Flow-CGnets well approximate the folding free energy landscape}
Figure~\ref{fig:fastfolders}e shows the joint densities over the first two TICA coordinates~\cite{Naritomi2011,Perez_JChemPhys2013,Schwantes_JChemTheoryComput2013}, see SI-Section~D.5. These reaction coordinates visualize the slowest processes in the MD simulation, which correspond to folding and unfolding, see SI for details.
The Flow-CGnet densities resemble the atomistic densities, showing that the global patterns in the folding process are captured. The match deteriorates with increasing sequence length:~for chignolin the Flow-CGnet recovers the shape of the distribution well, for trpcage and bba some minor metastable states are missing, and for villin some regions that are sparsely populated in the MD data are overstabilized.

\section{Discussion}
\subsection{Training data requirements}
Flow matching does not require the forces to be saved with the simulation data, and is thus more readily applicable than force-matching.
We have also shown that matching the empirical distribution benefits data efficiency. A drawback is that flow matching requires the underlying all-atom data to come from an equilibrated ensemble. However, this does not need to be achieved in long simulation trajectories:~reweighting from biased ensembles, such as replica-exchange simulations, or reweighting of short trajectories via Markov state models~\cite{Prinz_JChemPhys2011, Husic_JACS2018} are possible.

There are also theoretical developments in generalizing the force-matching method for non-equilibrium cases, such as Ref.~\citenum{harmandaris2016path}.
In such situations (but generally whenever atomistic force information is available), it might be beneficial to train the flow by combining density estimation with force-matching. Such a mixed loss can be especially efficient when using flows with continuous forces~\cite{kohler2021smooth}.

\subsection{Architectural choices for neural networks}
The teacher neural network needs to: (i) be trainable via (approximate) likelihood on sampling data, (ii) permit efficient sampling, and (iii) allow us to compute the instantaneous forces (Eq.~\eqref{eq:instantaneous-forces}).
We found that smooth mixture flows~\cite{kohler2021smooth} on the internal coordinates are able to reproduce the CG conformational distribution very accurately. 
Other latent variable models, including different normalizing flow architectures as well as variational autoencoders~\cite{Kingma2014} and their generalizations~\cite{wu2020snf, nielsen2020survae}, could be used as well. %
Examples of other generative networks used in coarse-graining applications can be found in Refs.~\citenum{Wang2019, DBLP:conf/icml/WangXCMSW0G22, Stieffenhofer2020}.

The student neural network is trained to represent the CG free energy. While here we used a modified version of the simple CGnet method~\cite{wang2019machine}, this network could be replaced by more advanced neural network architectures, such as SchNet~\cite{Schutt_JChemPhys2018}, other graph neural networks~\cite{husic2020coarse,Kovacs2021,Thomas2018,Klicpera2020,Batzner2021,Unke2019,Smith2017,Chen2021}, or other machine learning methods~\cite{John2017},
In principle, flow-matching can be combined with any trainable CG model, either based on neural networks or fixed functional forms with adjustable parameters.

\subsection{Scalability to larger molecules}
We observed that the CG model quality deteriorated and eventually became unusable for larger proteins.
This is because the present normalizing flows are built on a global internal coordinate representation. As the length of the peptide chain grows, the target potential energy becomes extremely sensitive with respect to these internal coordinates. For example, a tiny rotation of one torsion can easily cause steric clashes in a different part of the molecule. 
This may lead to, for example, a sigificant decrease of effective size of the training set after repulsion-based reweighting (see Fig.~S4).
Other work~\cite{noe2019boltzmann, wirnsberger2020targeted, Ding2021, Li2020, Gabrie2022} has also found suitable flow architectures for small molecules, proteins, and even explicitly solvated systems, but did not report whether they could produce quantitatively matching forces. 
One possibility to scale to large molecules is to employ coupling flows with equivariant neural networks operating in Cartesian space while still informed by internal coordinates, but further work is needed in order to find suitable architectures that can sample low-energy structures and have the relevant physics built in.

\section{Conclusions}
We have developed a two-stage approach to bottom-up coarse-graining that addresses two major problems with classical approaches, namely data availability and efficiency. The \emph{flow-matching} method produces thermodynamically consistent CG models without relying on either all-atom ground truth forces or subsequent CG simulations.
The key ingredient of our method is a generative deep neural network that is introduced into the optimization workflow. 
Compared with classical force-matching, flow-matching combined with CGnet captures the global thermodynamics of small peptides much more accurately than CGnet models trained with force-matching. Interestingly, this was even the case, when only a fraction ($<10\%$) of the data was used during training.
The main factor determining the data efficiency of flow-matching with respect to force-matching is the ratio of the number of atoms versus the number of CG particles. For the examples described in this work---where the all-atom systems are solvated macromolecules and the CG models retain only a few solute atoms---this ratio is very large, and the instantaneous all-atom forces projected on the CG coordinates are very noisy.

Applications to four small proteins yielded CG potentials that were able to fold and unfold the proteins via the same pathways as all-atom MD.
Biopolymers such as proteins are an especially interesting candidate for our proposed method, because they can be extremely difficult to sample, which makes the speedup obtained by a CG force field more practically attractive.
Furthermore, bottom-up coarse-graining in the present manner is applicable to many other molecular systems, including other polymers, liquids and materials.
Thereby, the present work opens a new and efficient path to reach near-atomistic accuracy on scales not amenable to atomistic simulations.

The two-step machine learning architecture consisting of a teacher and a student model gives rise to an interesting strategy for training transferable CG potentials: 
One may train separate system-specific teacher networks (e.g., flows) and then train a shared CG force field to obtain a transferable molecular model across the chemical space represented by the training data. 
Again, biopolymers are particularly interesting candidates for transferable CG force fields, as they usually consist of relatively few chemical building blocks which simplifies the parameterization of a force-field that can generalize across all sequences.
We envisage that flow-matching will be an important contribution to the development of transferable CG force fields and thereby help us to access time- and length-scales currently inaccessible to accurate molecular models. 

\begin{acknowledgement}

The authors thank Aleksander E.~P.~Durumeric, Nicholas E.~Charron, Brooke E.~Husic, Klara Bonneau, Manuel Dibak, Leon Klein, Michele Invernizzi and Leon Sixt for insightful discussions.
We gratefully acknowledge funding from the European Commission (Grant No.~ERC CoG 772230 \textquotedblleft ScaleCell\textquotedblright), the International Max Planck Research School for Biology and Computation (IMPRS\textendash BAC), the BMBF (Berlin Institute for Learning and Data, BIFOLD), the Berlin Mathematics center MATH+~(AA1-6, EF1-2) and the Deutsche Forschungsgemeinschaft DFG (GRK DAEDALUS, SFB1114/A04 and B08). 
C.C. acknowledges funding from the Deutsche Forschungsgemeinschaft DFG (SFB/TRR 186, Project A12; SFB 1114, Projects B03 and A04; SFB 1078, Project C7; and RTG 2433, Project Q05), the National Science Foundation (CHE-1900374, and PHY-2019745), and the Einstein Foundation Berlin (Project 0420815101).
The 3D molecular structures are visualized with PyMOL~\cite{PyMOL}.

\end{acknowledgement}

\begin{suppinfo}
The supporting information document starts with a concise summary of the methods applied for the training, validation and sampling of all coarse grained models in the main text. 
The following three sections cover comprehensive theoretical derivations (Section B), all necessary details of the coarse grained flow and CGnet model training and validation (Section C) and of the sample analyses and comparisons (Section D).

All CG samples from flow, Flow-CGnet and conventional CGnet models involved in the analyses in the main text have been deposited on the Zenodo platform with DOI:~\href{https://doi.org/10.5281/zenodo.7092156}{10.5281/zenodo.7092156}. The accompanying code for flow-matching and application examples are available at~\href{https://github.com/noegroup/flowm}{https://github.com/noegroup/flowm}.

\end{suppinfo}


\begin{mcitethebibliography}{78}
\providecommand*\natexlab[1]{#1}
\providecommand*\mciteSetBstSublistMode[1]{}
\providecommand*\mciteSetBstMaxWidthForm[2]{}
\providecommand*\mciteBstWouldAddEndPuncttrue
  {\def\EndOfBibitem{\unskip.}}
\providecommand*\mciteBstWouldAddEndPunctfalse
  {\let\EndOfBibitem\relax}
\providecommand*\mciteSetBstMidEndSepPunct[3]{}
\providecommand*\mciteSetBstSublistLabelBeginEnd[3]{}
\providecommand*\EndOfBibitem{}
\mciteSetBstSublistMode{f}
\mciteSetBstMaxWidthForm{subitem}{(\alph{mcitesubitemcount})}
\mciteSetBstSublistLabelBeginEnd
  {\mcitemaxwidthsubitemform\space}
  {\relax}
  {\relax}

\bibitem[Shaw \latin{et~al.}(2014)Shaw, Grossman, Bank, Batson, Butts, Chao,
  Deneroff, Dror, Even, Fenton, Forte, Gagliardo, Gill, Greskamp, Ho, Ierardi,
  Iserovich, Kuskin, Larson, Layman, Lee, Lerer, Li, Killebrew, Mackenzie, Mok,
  Moraes, Mueller, Nociolo, Peticolas, Quan, Ramot, Salmon, Scarpazza, {Ben
  Schafer}, Siddique, Snyder, Spengler, Tang, Theobald, Toma, Towles, Vitale,
  Wang, and Young]{Shaw2014}
Shaw,~D.~E.; Grossman,~J.~P.; Bank,~J.~A.; Batson,~B.; Butts,~J.~A.;
  Chao,~J.~C.; Deneroff,~M.~M.; Dror,~R.~O.; Even,~A.; Fenton,~C.~H.;
  Forte,~A.; Gagliardo,~J.; Gill,~G.; Greskamp,~B.; Ho,~C.~R.; Ierardi,~D.~J.;
  Iserovich,~L.; Kuskin,~J.~S.; Larson,~R.~H.; Layman,~T.; Lee,~L.~S.;
  Lerer,~A.~K.; Li,~C.; Killebrew,~D.; Mackenzie,~K.~M.; Mok,~S. Y.~H.;
  Moraes,~M.~A.; Mueller,~R.; Nociolo,~L.~J.; Peticolas,~J.~L.; Quan,~T.;
  Ramot,~D.; Salmon,~J.~K.; Scarpazza,~D.~P.; {Ben Schafer},~U.; Siddique,~N.;
  Snyder,~C.~W.; Spengler,~J.; Tang,~P. T.~P.; Theobald,~M.; Toma,~H.;
  Towles,~B.; Vitale,~B.; Wang,~S.~C.; Young,~C. {Anton 2: Raising the Bar for
  Performance and Programmability in a Special-Purpose Molecular Dynamics
  Supercomputer}. \emph{Int. Conf. High Perform. Comput. Networking, Storage
  Anal. SC} \textbf{2014}, \emph{2015-January}, 41--53\relax
\mciteBstWouldAddEndPuncttrue
\mciteSetBstMidEndSepPunct{\mcitedefaultmidpunct}
{\mcitedefaultendpunct}{\mcitedefaultseppunct}\relax
\EndOfBibitem
\bibitem[Shaw \latin{et~al.}(2021)Shaw, Adams, Azaria, Bank, Batson, Bell,
  Bergdorf, Bhatt, {Adam Butts}, Correi, Dirks, Dror, Eastwoo, Edwards, Even,
  Feldmann, Fenn, Fenton, Forte, Gagliardo, Gill, Gorlatova, Greskamp,
  Grossman, Gullingsrud, Harper, Hasenplaugh, Heily, Heshmat, Hunt, Ierardi,
  Iserovich, Jackson, Johnson, Kirk, Klepeis, Kuskin, Mackenzie, Mader,
  McGowen, McLaughlin, Moraes, Nasr, Nociolo, O'Donnell, Parker, Peticolas,
  Pocina, Predescu, Quan, Salmon, Schwink, Shim, Siddique, Spengler, Szalay,
  Tabladillo, Tartler, Taube, Theobald, Towles, Vick, Wang, Wazlowski,
  Weingarten, Williams, and Yuh]{Shaw2021}
Shaw,~D.~E.; Adams,~P.~J.; Azaria,~A.; Bank,~J.~A.; Batson,~B.; Bell,~A.;
  Bergdorf,~M.; Bhatt,~J.; {Adam Butts},~J.; Correi,~T.; Dirks,~R.~M.;
  Dror,~R.~O.; Eastwoo,~M.~P.; Edwards,~B.; Even,~A.; Feldmann,~P.; Fenn,~M.;
  Fenton,~C.~H.; Forte,~A.; Gagliardo,~J.; Gill,~G.; Gorlatova,~M.;
  Greskamp,~B.; Grossman,~J.~P.; Gullingsrud,~J.; Harper,~A.; Hasenplaugh,~W.;
  Heily,~M.; Heshmat,~B.~C.; Hunt,~J.; Ierardi,~D.~J.; Iserovich,~L.;
  Jackson,~B.~L.; Johnson,~N.~P.; Kirk,~M.~M.; Klepeis,~J.~L.; Kuskin,~J.~S.;
  Mackenzie,~K.~M.; Mader,~R.~J.; McGowen,~R.; McLaughlin,~A.; Moraes,~M.~A.;
  Nasr,~M.~H.; Nociolo,~L.~J.; O'Donnell,~L.; Parker,~A.; Peticolas,~J.~L.;
  Pocina,~G.; Predescu,~C.; Quan,~T.; Salmon,~J.~K.; Schwink,~C.; Shim,~K.~S.;
  Siddique,~N.; Spengler,~J.; Szalay,~T.; Tabladillo,~R.; Tartler,~R.;
  Taube,~A.~G.; Theobald,~M.; Towles,~B.; Vick,~W.; Wang,~S.~C.; Wazlowski,~M.;
  Weingarten,~M.~J.; Williams,~J.~M.; Yuh,~K.~A. {Anton 3: Twenty Microseconds
  of Molecular Dynamics Simulation before Lunch}. \emph{Int. Conf. High
  Perform. Comput. Networking, Storage Anal. SC} \textbf{2021}, \relax
\mciteBstWouldAddEndPunctfalse
\mciteSetBstMidEndSepPunct{\mcitedefaultmidpunct}
{}{\mcitedefaultseppunct}\relax
\EndOfBibitem
\bibitem[Prinz \latin{et~al.}(2011)Prinz, Wu, Sarich, Keller, Senne, Held,
  Chodera, Sch{\"u}tte, and No{\'e}]{Prinz_JChemPhys2011}
Prinz,~J.-H.; Wu,~H.; Sarich,~M.; Keller,~B.; Senne,~M.; Held,~M.;
  Chodera,~J.~D.; Sch{\"u}tte,~C.; No{\'e},~F. Markov models of molecular
  kinetics: Generation and validation. \emph{J. Chem. Phys.} \textbf{2011},
  \emph{134}, 174105\relax
\mciteBstWouldAddEndPuncttrue
\mciteSetBstMidEndSepPunct{\mcitedefaultmidpunct}
{\mcitedefaultendpunct}{\mcitedefaultseppunct}\relax
\EndOfBibitem
\bibitem[Husic and Pande(2018)Husic, and Pande]{Husic_JACS2018}
Husic,~B.~E.; Pande,~V.~S. Markov state models: From an art to a science.
  \emph{J. Am. Chem. Soc.} \textbf{2018}, \emph{140}, 2386--2396\relax
\mciteBstWouldAddEndPuncttrue
\mciteSetBstMidEndSepPunct{\mcitedefaultmidpunct}
{\mcitedefaultendpunct}{\mcitedefaultseppunct}\relax
\EndOfBibitem
\bibitem[Lindorff-Larsen \latin{et~al.}(2011)Lindorff-Larsen, Piana, Dror, and
  Shaw]{Lindorff_Science2011}
Lindorff-Larsen,~K.; Piana,~S.; Dror,~R.~O.; Shaw,~D.~E. How fast-folding
  proteins fold. \emph{Science} \textbf{2011}, \emph{334}, 517--520\relax
\mciteBstWouldAddEndPuncttrue
\mciteSetBstMidEndSepPunct{\mcitedefaultmidpunct}
{\mcitedefaultendpunct}{\mcitedefaultseppunct}\relax
\EndOfBibitem
\bibitem[Plattner \latin{et~al.}(2017)Plattner, Doerr, {De Fabritiis}, and
  No{\'{e}}]{Plattner2017}
Plattner,~N.; Doerr,~S.; {De Fabritiis},~G.; No{\'{e}},~F. {Complete
  protein–protein association kinetics in atomic detail revealed by molecular
  dynamics simulations and Markov modelling}. \emph{Nat. Chem. 2017 910}
  \textbf{2017}, \emph{9}, 1005--1011\relax
\mciteBstWouldAddEndPuncttrue
\mciteSetBstMidEndSepPunct{\mcitedefaultmidpunct}
{\mcitedefaultendpunct}{\mcitedefaultseppunct}\relax
\EndOfBibitem
\bibitem[Clementi \latin{et~al.}(2000)Clementi, Nymeyer, and
  Onuchic]{Clementi_JMolBiol2000}
Clementi,~C.; Nymeyer,~H.; Onuchic,~J.~N. Topological and energetic factors:
  what determines the structural details of the transition state ensemble and
  ``en-route'' intermediates for protein folding? An investigation for small
  globular proteins. \emph{J. Mol. Biol.} \textbf{2000}, \emph{298},
  937--953\relax
\mciteBstWouldAddEndPuncttrue
\mciteSetBstMidEndSepPunct{\mcitedefaultmidpunct}
{\mcitedefaultendpunct}{\mcitedefaultseppunct}\relax
\EndOfBibitem
\bibitem[Clementi(2008)]{ClementiCOSB}
Clementi,~C. Coarse-grained models of protein folding: Toy-models or predictive
  tools? \emph{Curr. Opin. Struct. Biol.} \textbf{2008}, \emph{18},
  10--15\relax
\mciteBstWouldAddEndPuncttrue
\mciteSetBstMidEndSepPunct{\mcitedefaultmidpunct}
{\mcitedefaultendpunct}{\mcitedefaultseppunct}\relax
\EndOfBibitem
\bibitem[Matysiak and Clementi(2004)Matysiak, and Clementi]{Matysiak2004}
Matysiak,~S.; Clementi,~C. {Optimal combination of theory and experiment for
  the characterization of the protein folding landscape of S6: How far can a
  minimalist model go?} \emph{J. Mol. Biol.} \textbf{2004}, \emph{343},
  235--248\relax
\mciteBstWouldAddEndPuncttrue
\mciteSetBstMidEndSepPunct{\mcitedefaultmidpunct}
{\mcitedefaultendpunct}{\mcitedefaultseppunct}\relax
\EndOfBibitem
\bibitem[Matysiak and Clementi(2006)Matysiak, and Clementi]{Matysiak2006}
Matysiak,~S.; Clementi,~C. {Minimalist Protein Model as a Diagnostic Tool for
  Misfolding and Aggregation}. \emph{J. Mol. Biol.} \textbf{2006}, \emph{363},
  297--308\relax
\mciteBstWouldAddEndPuncttrue
\mciteSetBstMidEndSepPunct{\mcitedefaultmidpunct}
{\mcitedefaultendpunct}{\mcitedefaultseppunct}\relax
\EndOfBibitem
\bibitem[Das \latin{et~al.}(2005)Das, Matysiak, and Clementi]{das2005balancing}
Das,~P.; Matysiak,~S.; Clementi,~C. Balancing energy and entropy: A minimalist
  model for the characterization of protein folding landscapes. \emph{Proc.
  Natl. Acad. Sci. USA} \textbf{2005}, \emph{102}, 10141--10146\relax
\mciteBstWouldAddEndPuncttrue
\mciteSetBstMidEndSepPunct{\mcitedefaultmidpunct}
{\mcitedefaultendpunct}{\mcitedefaultseppunct}\relax
\EndOfBibitem
\bibitem[Saunders and Voth(2013)Saunders, and
  Voth]{Saunders_AnnuRevBiophys2013}
Saunders,~M.~G.; Voth,~G.~A. {Coarse-Graining Methods for Computational
  Biology}. \emph{Annu. Rev. Bioph. Biom.} \textbf{2013}, \emph{42},
  73--93\relax
\mciteBstWouldAddEndPuncttrue
\mciteSetBstMidEndSepPunct{\mcitedefaultmidpunct}
{\mcitedefaultendpunct}{\mcitedefaultseppunct}\relax
\EndOfBibitem
\bibitem[Noid(2013)]{Noid_JChemPhys2013}
Noid,~W.~G. {Perspective: Coarse-grained models for biomolecular systems}.
  \emph{J. Chem. Phys.} \textbf{2013}, \emph{139}, 090901\relax
\mciteBstWouldAddEndPuncttrue
\mciteSetBstMidEndSepPunct{\mcitedefaultmidpunct}
{\mcitedefaultendpunct}{\mcitedefaultseppunct}\relax
\EndOfBibitem
\bibitem[Ing\'{o}lfsson \latin{et~al.}(2014)Ing\'{o}lfsson, Lopez, Uusitalo,
  de~Jong, Gopal, Periole, and Marrink]{Ingolfsson_WIREsComputMolSci2014}
Ing\'{o}lfsson,~H.~I.; Lopez,~C.~A.; Uusitalo,~J.~J.; de~Jong,~D.~H.;
  Gopal,~S.~M.; Periole,~X.; Marrink,~S.~J. {The power of coarse graining in
  biomolecular simulations}. \emph{WIREs Comput. Mol. Sci.} \textbf{2014},
  \emph{4}, 225--248\relax
\mciteBstWouldAddEndPuncttrue
\mciteSetBstMidEndSepPunct{\mcitedefaultmidpunct}
{\mcitedefaultendpunct}{\mcitedefaultseppunct}\relax
\EndOfBibitem
\bibitem[Kmiecik \latin{et~al.}(2016)Kmiecik, Gront, Kolinski, Wieteska, Dawid,
  and Kolinski]{Kmiecik_ChemRev2016}
Kmiecik,~S.; Gront,~D.; Kolinski,~M.; Wieteska,~L.; Dawid,~A.~E.; Kolinski,~A.
  Coarse-grained protein models and their applications. \emph{Chem. Rev.}
  \textbf{2016}, \emph{116}, 7898--7936\relax
\mciteBstWouldAddEndPuncttrue
\mciteSetBstMidEndSepPunct{\mcitedefaultmidpunct}
{\mcitedefaultendpunct}{\mcitedefaultseppunct}\relax
\EndOfBibitem
\bibitem[Pak and Voth(2018)Pak, and Voth]{Pak_CurrOpinStructBiol2018}
Pak,~A.~J.; Voth,~G.~A. {Advances in coarse-grained modeling of macromolecular
  complexes}. \emph{Curr. Opin. Struc. Biol.} \textbf{2018}, \emph{52},
  119\relax
\mciteBstWouldAddEndPuncttrue
\mciteSetBstMidEndSepPunct{\mcitedefaultmidpunct}
{\mcitedefaultendpunct}{\mcitedefaultseppunct}\relax
\EndOfBibitem
\bibitem[Chen \latin{et~al.}(2018)Chen, Chen, Pinamonti, and
  Clementi]{Chen_JChemTheoryComput2018}
Chen,~J.; Chen,~J.; Pinamonti,~G.; Clementi,~C. Learning effective molecular
  models from experimental observables. \emph{J. Chem. Theory Comput.}
  \textbf{2018}, \emph{14}, 3849--3858\relax
\mciteBstWouldAddEndPuncttrue
\mciteSetBstMidEndSepPunct{\mcitedefaultmidpunct}
{\mcitedefaultendpunct}{\mcitedefaultseppunct}\relax
\EndOfBibitem
\bibitem[Singh and Li(2019)Singh, and Li]{Singh_IntJMolSci2019}
Singh,~N.; Li,~W. {Recent Advances in Coarse-Grained Models for Biomolecules
  and Their Applications}. \emph{Int. J. Mol. Sci.} \textbf{2019},
  \emph{20}\relax
\mciteBstWouldAddEndPuncttrue
\mciteSetBstMidEndSepPunct{\mcitedefaultmidpunct}
{\mcitedefaultendpunct}{\mcitedefaultseppunct}\relax
\EndOfBibitem
\bibitem[N{\"u}ske \latin{et~al.}(2019)N{\"u}ske, Boninsegna, and
  Clementi]{Nuske_JChemPhys2019}
N{\"u}ske,~F.; Boninsegna,~L.; Clementi,~C. Coarse-graining molecular systems
  by spectral matching. \emph{J. Chem. Phys.} \textbf{2019}, \emph{151},
  044116\relax
\mciteBstWouldAddEndPuncttrue
\mciteSetBstMidEndSepPunct{\mcitedefaultmidpunct}
{\mcitedefaultendpunct}{\mcitedefaultseppunct}\relax
\EndOfBibitem
\bibitem[Wang \latin{et~al.}(2019)Wang, Olsson, Wehmeyer, P{\'e}rez, Charron,
  De~Fabritiis, No{\'e}, and Clementi]{wang2019machine}
Wang,~J.; Olsson,~S.; Wehmeyer,~C.; P{\'e}rez,~A.; Charron,~N.~E.;
  De~Fabritiis,~G.; No{\'e},~F.; Clementi,~C. Machine learning of
  coarse-grained molecular dynamics force fields. \emph{ACS central science}
  \textbf{2019}, \emph{5}, 755--767\relax
\mciteBstWouldAddEndPuncttrue
\mciteSetBstMidEndSepPunct{\mcitedefaultmidpunct}
{\mcitedefaultendpunct}{\mcitedefaultseppunct}\relax
\EndOfBibitem
\bibitem[Wang \latin{et~al.}(2020)Wang, Chmiela, M{\"u}ller, No{\'e}, and
  Clementi]{Wang_JChemPhys2020}
Wang,~J.; Chmiela,~S.; M{\"u}ller,~K.-R.; No{\'e},~F.; Clementi,~C. Ensemble
  learning of coarse-grained molecular dynamics force fields with a kernel
  approach. \emph{J. Chem. Phys.} \textbf{2020}, \emph{152}, 194106\relax
\mciteBstWouldAddEndPuncttrue
\mciteSetBstMidEndSepPunct{\mcitedefaultmidpunct}
{\mcitedefaultendpunct}{\mcitedefaultseppunct}\relax
\EndOfBibitem
\bibitem[Husic \latin{et~al.}(2020)Husic, Charron, Lemm, Wang, P{\'e}rez,
  Majewski, Kr{\"a}mer, Chen, Olsson, de~Fabritiis, \latin{et~al.}
  others]{husic2020coarse}
Husic,~B.~E.; Charron,~N.~E.; Lemm,~D.; Wang,~J.; P{\'e}rez,~A.; Majewski,~M.;
  Kr{\"a}mer,~A.; Chen,~Y.; Olsson,~S.; de~Fabritiis,~G., \latin{et~al.}
  Coarse graining molecular dynamics with graph neural networks. \emph{J. Chem.
  Phys.} \textbf{2020}, \emph{153}, 194101\relax
\mciteBstWouldAddEndPuncttrue
\mciteSetBstMidEndSepPunct{\mcitedefaultmidpunct}
{\mcitedefaultendpunct}{\mcitedefaultseppunct}\relax
\EndOfBibitem
\bibitem[Jin \latin{et~al.}(2022)Jin, Pak, Durumeric, Loose, and
  Voth]{jin2022bottom}
Jin,~J.; Pak,~A.~J.; Durumeric,~A.~E.; Loose,~T.~D.; Voth,~G.~A. Bottom-up
  Coarse-Graining: Principles and Perspectives. \emph{Journal of Chemical
  Theory and Computation} \textbf{2022}, \emph{18}, 5759--5791\relax
\mciteBstWouldAddEndPuncttrue
\mciteSetBstMidEndSepPunct{\mcitedefaultmidpunct}
{\mcitedefaultendpunct}{\mcitedefaultseppunct}\relax
\EndOfBibitem
\bibitem[Boninsegna \latin{et~al.}(2018)Boninsegna, Banisch, and
  Clementi]{boninsegna2018data}
Boninsegna,~L.; Banisch,~R.; Clementi,~C. A data-driven perspective on the
  hierarchical assembly of molecular structures. \emph{Journal of Chemical
  Theory and Computation} \textbf{2018}, \emph{14}, 453--460\relax
\mciteBstWouldAddEndPuncttrue
\mciteSetBstMidEndSepPunct{\mcitedefaultmidpunct}
{\mcitedefaultendpunct}{\mcitedefaultseppunct}\relax
\EndOfBibitem
\bibitem[Wang and G{\'{o}}mez-Bombarelli(2019)Wang, and
  G{\'{o}}mez-Bombarelli]{Wang2019}
Wang,~W.; G{\'{o}}mez-Bombarelli,~R. {Coarse-graining auto-encoders for
  molecular dynamics}. \emph{npj Comput. Mater.} \textbf{2019}, \emph{5},
  1--9\relax
\mciteBstWouldAddEndPuncttrue
\mciteSetBstMidEndSepPunct{\mcitedefaultmidpunct}
{\mcitedefaultendpunct}{\mcitedefaultseppunct}\relax
\EndOfBibitem
\bibitem[Wagner \latin{et~al.}(2016)Wagner, Dama, Durumeric, and
  Voth]{wagner2016representability}
Wagner,~J.~W.; Dama,~J.~F.; Durumeric,~A.~E.; Voth,~G.~A. On the
  representability problem and the physical meaning of coarse-grained models.
  \emph{J. Chem. Phys.} \textbf{2016}, \emph{145}, 044108\relax
\mciteBstWouldAddEndPuncttrue
\mciteSetBstMidEndSepPunct{\mcitedefaultmidpunct}
{\mcitedefaultendpunct}{\mcitedefaultseppunct}\relax
\EndOfBibitem
\bibitem[Dunn \latin{et~al.}(2016)Dunn, Foley, and Noid]{dunn2016van}
Dunn,~N.~J.; Foley,~T.~T.; Noid,~W.~G. Van der Waals perspective on
  coarse-graining: Progress toward solving representability and transferability
  problems. \emph{Acc. Chem. Res.} \textbf{2016}, \emph{49}, 2832--2840\relax
\mciteBstWouldAddEndPuncttrue
\mciteSetBstMidEndSepPunct{\mcitedefaultmidpunct}
{\mcitedefaultendpunct}{\mcitedefaultseppunct}\relax
\EndOfBibitem
\bibitem[Jin \latin{et~al.}(2019)Jin, Pak, and Voth]{jin2019understanding}
Jin,~J.; Pak,~A.~J.; Voth,~G.~A. Understanding missing entropy in
  coarse-grained systems: Addressing issues of representability and
  transferability. \emph{J. Phys. Chem. Lett.} \textbf{2019}, \emph{10},
  4549--4557\relax
\mciteBstWouldAddEndPuncttrue
\mciteSetBstMidEndSepPunct{\mcitedefaultmidpunct}
{\mcitedefaultendpunct}{\mcitedefaultseppunct}\relax
\EndOfBibitem
\bibitem[Dannenhoffer-Lafage \latin{et~al.}(2019)Dannenhoffer-Lafage, Wagner,
  Durumeric, and Voth]{dannenhoffer2019compatible}
Dannenhoffer-Lafage,~T.; Wagner,~J.~W.; Durumeric,~A.~E.; Voth,~G.~A.
  Compatible observable decompositions for coarse-grained representations of
  real molecular systems. \emph{J. Chem. Phys.} \textbf{2019}, \emph{151},
  134115\relax
\mciteBstWouldAddEndPuncttrue
\mciteSetBstMidEndSepPunct{\mcitedefaultmidpunct}
{\mcitedefaultendpunct}{\mcitedefaultseppunct}\relax
\EndOfBibitem
\bibitem[Lebold and Noid(2019)Lebold, and Noid]{doi:10.1063/1.5125246}
Lebold,~K.~M.; Noid,~W.~G. Dual-potential approach for coarse-grained implicit
  solvent models with accurate, internally consistent energetics and predictive
  transferability. \emph{The Journal of Chemical Physics} \textbf{2019},
  \emph{151}, 164113\relax
\mciteBstWouldAddEndPuncttrue
\mciteSetBstMidEndSepPunct{\mcitedefaultmidpunct}
{\mcitedefaultendpunct}{\mcitedefaultseppunct}\relax
\EndOfBibitem
\bibitem[Reith \latin{et~al.}(2003)Reith, Pütz, and Müller-Plathe]{Reith2003}
Reith,~D.; Pütz,~M.; Müller-Plathe,~F. Deriving effective mesoscale
  potentials from atomistic simulations. \emph{Journal of Computational
  Chemistry} \textbf{2003}, \emph{24}, 1624--1636\relax
\mciteBstWouldAddEndPuncttrue
\mciteSetBstMidEndSepPunct{\mcitedefaultmidpunct}
{\mcitedefaultendpunct}{\mcitedefaultseppunct}\relax
\EndOfBibitem
\bibitem[Izvekov and Voth(2005)Izvekov, and Voth]{Izvekov_JPhysChemB2005}
Izvekov,~S.; Voth,~G.~A. {A Multiscale Coarse-Graining Method for Biomolecular
  Systems}. \emph{J. Phys. Chem. B} \textbf{2005}, \emph{109}, 2469--2473\relax
\mciteBstWouldAddEndPuncttrue
\mciteSetBstMidEndSepPunct{\mcitedefaultmidpunct}
{\mcitedefaultendpunct}{\mcitedefaultseppunct}\relax
\EndOfBibitem
\bibitem[Noid \latin{et~al.}(2008)Noid, Chu, Ayton, Krishna, Izvekov, Voth,
  Das, and Andersen]{Noid_JChemPhys2008}
Noid,~W.~G.; Chu,~J.-W.; Ayton,~G.~S.; Krishna,~V.; Izvekov,~S.; Voth,~G.~A.;
  Das,~A.; Andersen,~H.~C. {The multiscale coarse-graining method. I. A
  rigorous bridge between atomistic and coarse-grained models}. \emph{J. Chem.
  Phys.} \textbf{2008}, \emph{128}, 244114\relax
\mciteBstWouldAddEndPuncttrue
\mciteSetBstMidEndSepPunct{\mcitedefaultmidpunct}
{\mcitedefaultendpunct}{\mcitedefaultseppunct}\relax
\EndOfBibitem
\bibitem[Shell(2008)]{Shell_JChemPhys2008}
Shell,~M.~S. The relative entropy is fundamental to multiscale and inverse
  thermodynamic problems. \emph{J. Chem. Phys.} \textbf{2008}, \emph{129},
  144108\relax
\mciteBstWouldAddEndPuncttrue
\mciteSetBstMidEndSepPunct{\mcitedefaultmidpunct}
{\mcitedefaultendpunct}{\mcitedefaultseppunct}\relax
\EndOfBibitem
\bibitem[Lyubartsev and Laaksonen(1995)Lyubartsev, and
  Laaksonen]{PhysRevE.52.3730}
Lyubartsev,~A.~P.; Laaksonen,~A. Calculation of effective interaction
  potentials from radial distribution functions: A reverse Monte Carlo
  approach. \emph{Phys. Rev. E} \textbf{1995}, \emph{52}, 3730--3737\relax
\mciteBstWouldAddEndPuncttrue
\mciteSetBstMidEndSepPunct{\mcitedefaultmidpunct}
{\mcitedefaultendpunct}{\mcitedefaultseppunct}\relax
\EndOfBibitem
\bibitem[Pak \latin{et~al.}(2019)Pak, Dannenhoffer-Lafage, Madsen, and
  Voth]{Pak2019}
Pak,~A.~J.; Dannenhoffer-Lafage,~T.; Madsen,~J.~J.; Voth,~G.~A. Systematic
  Coarse-Grained Lipid Force Fields with Semiexplicit Solvation via Virtual
  Sites. \emph{J. Chem. Theory Comput.} \textbf{2019}, \emph{15},
  2087--2100\relax
\mciteBstWouldAddEndPuncttrue
\mciteSetBstMidEndSepPunct{\mcitedefaultmidpunct}
{\mcitedefaultendpunct}{\mcitedefaultseppunct}\relax
\EndOfBibitem
\bibitem[Tabak \latin{et~al.}(2010)Tabak, Vanden-Eijnden, \latin{et~al.}
  others]{tabak2010density}
Tabak,~E.~G.; Vanden-Eijnden,~E., \latin{et~al.}  Density estimation by dual
  ascent of the log-likelihood. \emph{Communications in Mathematical Sciences}
  \textbf{2010}, \emph{8}, 217--233\relax
\mciteBstWouldAddEndPuncttrue
\mciteSetBstMidEndSepPunct{\mcitedefaultmidpunct}
{\mcitedefaultendpunct}{\mcitedefaultseppunct}\relax
\EndOfBibitem
\bibitem[Rezende and Mohamed(2015)Rezende, and Mohamed]{rezende2015variational}
Rezende,~D.; Mohamed,~S. Variational inference with normalizing flows.
  International Conference on Machine Learning. 2015; pp 1530--1538\relax
\mciteBstWouldAddEndPuncttrue
\mciteSetBstMidEndSepPunct{\mcitedefaultmidpunct}
{\mcitedefaultendpunct}{\mcitedefaultseppunct}\relax
\EndOfBibitem
\bibitem[Papamakarios \latin{et~al.}(2021)Papamakarios, Nalisnick, Rezende,
  Mohamed, and Lakshminarayanan]{papamakarios2019normalizing}
Papamakarios,~G.; Nalisnick,~E.; Rezende,~D.~J.; Mohamed,~S.;
  Lakshminarayanan,~B. Normalizing flows for probabilistic modeling and
  inference. \emph{Journal of Machine Learning Research} \textbf{2021},
  \emph{22}, 1--64\relax
\mciteBstWouldAddEndPuncttrue
\mciteSetBstMidEndSepPunct{\mcitedefaultmidpunct}
{\mcitedefaultendpunct}{\mcitedefaultseppunct}\relax
\EndOfBibitem
\bibitem[No{\'e} \latin{et~al.}(2019)No{\'e}, Olsson, K{\"o}hler, and
  Wu]{noe2019boltzmann}
No{\'e},~F.; Olsson,~S.; K{\"o}hler,~J.; Wu,~H. Boltzmann generators: Sampling
  equilibrium states of many-body systems with deep learning. \emph{Science}
  \textbf{2019}, \emph{365}, eaaw1147\relax
\mciteBstWouldAddEndPuncttrue
\mciteSetBstMidEndSepPunct{\mcitedefaultmidpunct}
{\mcitedefaultendpunct}{\mcitedefaultseppunct}\relax
\EndOfBibitem
\bibitem[Gabri{\'{e}}u \latin{et~al.}(2022)Gabri{\'{e}}u, Rotskoff, and
  Vanden-Eijnden]{Gabrie2022}
Gabri{\'{e}}u,~M.; Rotskoff,~G.~M.; Vanden-Eijnden,~E. Adaptive Monte Carlo
  augmented with normalizing flows. \emph{Proceedings of the National Academy
  of Sciences of the United States of America} \textbf{2022}, \emph{119}\relax
\mciteBstWouldAddEndPuncttrue
\mciteSetBstMidEndSepPunct{\mcitedefaultmidpunct}
{\mcitedefaultendpunct}{\mcitedefaultseppunct}\relax
\EndOfBibitem
\bibitem[Li \latin{et~al.}(2020)Li, Dong, Zhang, and Wang]{Li2020}
Li,~S.~H.; Dong,~C.~X.; Zhang,~L.; Wang,~L. {Neural Canonical Transformation
  with Symplectic Flows}. \emph{Phys. Rev. X} \textbf{2020}, \emph{10},
  021020\relax
\mciteBstWouldAddEndPuncttrue
\mciteSetBstMidEndSepPunct{\mcitedefaultmidpunct}
{\mcitedefaultendpunct}{\mcitedefaultseppunct}\relax
\EndOfBibitem
\bibitem[Nicoli \latin{et~al.}(2020)Nicoli, Nakajima, Strodthoff, Samek,
  M{\"{u}}ller, and Kessel]{Nicoli2020}
Nicoli,~K.~A.; Nakajima,~S.; Strodthoff,~N.; Samek,~W.; M{\"{u}}ller,~K.~R.;
  Kessel,~P. {Asymptotically unbiased estimation of physical observables with
  neural samplers}. \emph{Phys. Rev. E} \textbf{2020}, \emph{101}, 023304\relax
\mciteBstWouldAddEndPuncttrue
\mciteSetBstMidEndSepPunct{\mcitedefaultmidpunct}
{\mcitedefaultendpunct}{\mcitedefaultseppunct}\relax
\EndOfBibitem
\bibitem[Liu \latin{et~al.}(2021)Liu, Xu, Jiang, and Wong]{Liu2021}
Liu,~Q.; Xu,~J.; Jiang,~R.; Wong,~W.~H. {Density estimation using deep
  generative neural networks}. \emph{Proc. Natl. Acad. Sci. U. S. A.}
  \textbf{2021}, \emph{118}\relax
\mciteBstWouldAddEndPuncttrue
\mciteSetBstMidEndSepPunct{\mcitedefaultmidpunct}
{\mcitedefaultendpunct}{\mcitedefaultseppunct}\relax
\EndOfBibitem
\bibitem[Ding and Zhang(2021)Ding, and Zhang]{Ding2021}
Ding,~X.; Zhang,~B. {DeepBAR: A Fast and Exact Method for Binding Free Energy
  Computation}. \emph{J. Phys. Chem. Lett.} \textbf{2021}, \emph{12},
  2509--2515\relax
\mciteBstWouldAddEndPuncttrue
\mciteSetBstMidEndSepPunct{\mcitedefaultmidpunct}
{\mcitedefaultendpunct}{\mcitedefaultseppunct}\relax
\EndOfBibitem
\bibitem[Wirnsberger \latin{et~al.}(2020)Wirnsberger, Ballard, Papamakarios,
  Abercrombie, Racani{\`e}re, Pritzel, Blundell, \latin{et~al.}
  others]{wirnsberger2020targeted}
Wirnsberger,~P.; Ballard,~A.; Papamakarios,~G.; Abercrombie,~S.;
  Racani{\`e}re,~S.; Pritzel,~A.; Blundell,~C., \latin{et~al.}  Targeted free
  energy estimation via learned mappings. \emph{The Journal of Chemical
  Physics} \textbf{2020}, \emph{153}, 144112--144112\relax
\mciteBstWouldAddEndPuncttrue
\mciteSetBstMidEndSepPunct{\mcitedefaultmidpunct}
{\mcitedefaultendpunct}{\mcitedefaultseppunct}\relax
\EndOfBibitem
\bibitem[Wang \latin{et~al.}(2022)Wang, Xu, Cai, Miller, Smidt, Wang, Tang, and
  G{\'{o}}mez{-}Bombarelli]{DBLP:conf/icml/WangXCMSW0G22}
Wang,~W.; Xu,~M.; Cai,~C.; Miller,~B.~K.; Smidt,~T.~E.; Wang,~Y.; Tang,~J.;
  G{\'{o}}mez{-}Bombarelli,~R. Generative Coarse-Graining of Molecular
  Conformations. International Conference on Machine Learning, {ICML} 2022,
  17-23 July 2022, Baltimore, Maryland, {USA}. 2022; pp 23213--23236\relax
\mciteBstWouldAddEndPuncttrue
\mciteSetBstMidEndSepPunct{\mcitedefaultmidpunct}
{\mcitedefaultendpunct}{\mcitedefaultseppunct}\relax
\EndOfBibitem
\bibitem[Stieffenhofer \latin{et~al.}(2020)Stieffenhofer, Wand, and
  Bereau]{Stieffenhofer2020}
Stieffenhofer,~M.; Wand,~M.; Bereau,~T. {Adversarial reverse mapping of
  equilibrated condensed-phase molecular structures}. \emph{Mach. Learn. Sci.
  Technol.} \textbf{2020}, \emph{1}, 045014\relax
\mciteBstWouldAddEndPuncttrue
\mciteSetBstMidEndSepPunct{\mcitedefaultmidpunct}
{\mcitedefaultendpunct}{\mcitedefaultseppunct}\relax
\EndOfBibitem
\bibitem[Mullinax and Noid(2009)Mullinax, and Noid]{Mullinax2009}
Mullinax,~J.~W.; Noid,~W.~G. Generalized Yvon-Born-Green Theory for Molecular
  Systems. \emph{Phys. Rev. Lett.} \textbf{2009}, \emph{103}\relax
\mciteBstWouldAddEndPuncttrue
\mciteSetBstMidEndSepPunct{\mcitedefaultmidpunct}
{\mcitedefaultendpunct}{\mcitedefaultseppunct}\relax
\EndOfBibitem
\bibitem[Wu \latin{et~al.}(2020)Wu, K\"{o}hler, and Noe]{wu2020snf}
Wu,~H.; K\"{o}hler,~J.; Noe,~F. Stochastic Normalizing Flows. Advances in
  Neural Information Processing Systems. 2020; pp 5933--5944\relax
\mciteBstWouldAddEndPuncttrue
\mciteSetBstMidEndSepPunct{\mcitedefaultmidpunct}
{\mcitedefaultendpunct}{\mcitedefaultseppunct}\relax
\EndOfBibitem
\bibitem[Ciccotti \latin{et~al.}(2008)Ciccotti, Lelievre, and
  Vanden-Eijnden]{Ciccotti_CommunPureApplMath2008}
Ciccotti,~G.; Lelievre,~T.; Vanden-Eijnden,~E. Projection of diffusions on
  submanifolds: Application to mean force computation. \emph{Commun. Pure Appl.
  Math.} \textbf{2008}, \emph{61}, 371--408\relax
\mciteBstWouldAddEndPuncttrue
\mciteSetBstMidEndSepPunct{\mcitedefaultmidpunct}
{\mcitedefaultendpunct}{\mcitedefaultseppunct}\relax
\EndOfBibitem
\bibitem[Kalligiannaki \latin{et~al.}(2015)Kalligiannaki, Harmandaris,
  Katsoulakis, and Plech{\'{a}}{\v{c}}]{Kalligiannaki2015}
Kalligiannaki,~E.; Harmandaris,~V.; Katsoulakis,~M.~A.; Plech{\'{a}}{\v{c}},~P.
  The geometry of generalized force matching and related information metrics in
  coarse-graining of molecular systems. \emph{J. Chem. Phys.} \textbf{2015},
  \emph{143}, 084105\relax
\mciteBstWouldAddEndPuncttrue
\mciteSetBstMidEndSepPunct{\mcitedefaultmidpunct}
{\mcitedefaultendpunct}{\mcitedefaultseppunct}\relax
\EndOfBibitem
\bibitem[Davtyan \latin{et~al.}(2016)Davtyan, Voth, and Andersen]{Davtyan2016}
Davtyan,~A.; Voth,~G.~A.; Andersen,~H.~C. Dynamic force matching: Construction
  of dynamic coarse-grained models with realistic short time dynamics and
  accurate long time dynamics. \emph{J. Chem. Phys.} \textbf{2016}, \emph{145},
  224107\relax
\mciteBstWouldAddEndPuncttrue
\mciteSetBstMidEndSepPunct{\mcitedefaultmidpunct}
{\mcitedefaultendpunct}{\mcitedefaultseppunct}\relax
\EndOfBibitem
\bibitem[LeCun \latin{et~al.}(2007)LeCun, Chopra, Hadsell, Ranzato, and
  Huang]{lecun2007}
LeCun,~Y.; Chopra,~S.; Hadsell,~R.; Ranzato,~M.; Huang,~F. \emph{{Predicting
  Structured Data}}; The MIT Press, 2007\relax
\mciteBstWouldAddEndPuncttrue
\mciteSetBstMidEndSepPunct{\mcitedefaultmidpunct}
{\mcitedefaultendpunct}{\mcitedefaultseppunct}\relax
\EndOfBibitem
\bibitem[Huang \latin{et~al.}(2020)Huang, Dinh, and
  Courville]{huang2020augmented}
Huang,~C.-W.; Dinh,~L.; Courville,~A. Augmented normalizing flows: Bridging the
  gap between generative flows and latent variable models. \emph{arXiv preprint
  arXiv:2002.07101} \textbf{2020}, \relax
\mciteBstWouldAddEndPunctfalse
\mciteSetBstMidEndSepPunct{\mcitedefaultmidpunct}
{}{\mcitedefaultseppunct}\relax
\EndOfBibitem
\bibitem[Chen \latin{et~al.}(2020)Chen, Lu, Chenli, Zhu, and
  Tian]{chen2020vflow}
Chen,~J.; Lu,~C.; Chenli,~B.; Zhu,~J.; Tian,~T. Vflow: More expressive
  generative flows with variational data augmentation. International Conference
  on Machine Learning. 2020; pp 1660--1669\relax
\mciteBstWouldAddEndPuncttrue
\mciteSetBstMidEndSepPunct{\mcitedefaultmidpunct}
{\mcitedefaultendpunct}{\mcitedefaultseppunct}\relax
\EndOfBibitem
\bibitem[Cornish \latin{et~al.}(2020)Cornish, Caterini, Deligiannidis, and
  Doucet]{cornish2020relaxing}
Cornish,~R.; Caterini,~A.; Deligiannidis,~G.; Doucet,~A. Relaxing bijectivity
  constraints with continuously indexed normalising flows. International
  Conference on Machine Learning. 2020; pp 2133--2143\relax
\mciteBstWouldAddEndPuncttrue
\mciteSetBstMidEndSepPunct{\mcitedefaultmidpunct}
{\mcitedefaultendpunct}{\mcitedefaultseppunct}\relax
\EndOfBibitem
\bibitem[Brofos \latin{et~al.}(2021)Brofos, Brubaker, and
  Lederman]{brofos2021manifold}
Brofos,~J.; Brubaker,~M.~A.; Lederman,~R.~R. Manifold Density Estimation via
  Generalized Dequantization. ICML Workshop on Invertible Neural Networks,
  Normalizing Flows, and Explicit Likelihood Models. 2021\relax
\mciteBstWouldAddEndPuncttrue
\mciteSetBstMidEndSepPunct{\mcitedefaultmidpunct}
{\mcitedefaultendpunct}{\mcitedefaultseppunct}\relax
\EndOfBibitem
\bibitem[Tobias and Brooks~III(1992)Tobias, and
  Brooks~III]{Tobias_JPhysChem1992}
Tobias,~D.~J.; Brooks~III,~C.~L. Conformational equilibrium in the alanine
  dipeptide in the gas phase and aqueous solution: A comparison of theoretical
  results. \emph{J. Phys. Chem.} \textbf{1992}, \emph{96}, 3864--3870\relax
\mciteBstWouldAddEndPuncttrue
\mciteSetBstMidEndSepPunct{\mcitedefaultmidpunct}
{\mcitedefaultendpunct}{\mcitedefaultseppunct}\relax
\EndOfBibitem
\bibitem[Clementi \latin{et~al.}(2003)Clementi, Garc{\i}a, and
  Onuchic]{Clementi2003}
Clementi,~C.; Garc{\i}a,~A.~E.; Onuchic,~J.~N. Interplay Among Tertiary
  Contacts, Secondary Structure Formation and Side-chain Packing in the Protein
  Folding Mechanism: All-atom Representation Study of Protein L. \emph{J. Mol.
  Biol.} \textbf{2003}, \emph{326}, 933--954\relax
\mciteBstWouldAddEndPuncttrue
\mciteSetBstMidEndSepPunct{\mcitedefaultmidpunct}
{\mcitedefaultendpunct}{\mcitedefaultseppunct}\relax
\EndOfBibitem
\bibitem[Naritomi and Fuchigami(2011)Naritomi, and Fuchigami]{Naritomi2011}
Naritomi,~Y.; Fuchigami,~S. Slow dynamics in protein fluctuations revealed by
  time-structure based independent component analysis: The case of domain
  motions. \emph{The Journal of Chemical Physics} \textbf{2011}, \emph{134},
  065101, TICA pioneer 3/3\relax
\mciteBstWouldAddEndPuncttrue
\mciteSetBstMidEndSepPunct{\mcitedefaultmidpunct}
{\mcitedefaultendpunct}{\mcitedefaultseppunct}\relax
\EndOfBibitem
\bibitem[P{\'e}rez-Hern{\'a}ndez \latin{et~al.}(2013)P{\'e}rez-Hern{\'a}ndez,
  Paul, Giorgino, De~Fabritiis, and No{\'e}]{Perez_JChemPhys2013}
P{\'e}rez-Hern{\'a}ndez,~G.; Paul,~F.; Giorgino,~T.; De~Fabritiis,~G.;
  No{\'e},~F. Identification of slow molecular order parameters for Markov
  model construction. \emph{J. Chem. Phys.} \textbf{2013}, \emph{139},
  07B604\_1\relax
\mciteBstWouldAddEndPuncttrue
\mciteSetBstMidEndSepPunct{\mcitedefaultmidpunct}
{\mcitedefaultendpunct}{\mcitedefaultseppunct}\relax
\EndOfBibitem
\bibitem[Schwantes and Pande(2013)Schwantes, and
  Pande]{Schwantes_JChemTheoryComput2013}
Schwantes,~C.~R.; Pande,~V.~S. Improvements in Markov state model construction
  reveal many non-native interactions in the folding of NTL9. \emph{J. Chem.
  Theory Comput.} \textbf{2013}, \emph{9}, 2000--2009\relax
\mciteBstWouldAddEndPuncttrue
\mciteSetBstMidEndSepPunct{\mcitedefaultmidpunct}
{\mcitedefaultendpunct}{\mcitedefaultseppunct}\relax
\EndOfBibitem
\bibitem[Harmandaris \latin{et~al.}(2016)Harmandaris, Kalligiannaki,
  Katsoulakis, and Plech{\'a}{\v{c}}]{harmandaris2016path}
Harmandaris,~V.; Kalligiannaki,~E.; Katsoulakis,~M.; Plech{\'a}{\v{c}},~P.
  Path-space variational inference for non-equilibrium coarse-grained systems.
  \emph{J. Comput. Phys.} \textbf{2016}, \emph{314}, 355--383\relax
\mciteBstWouldAddEndPuncttrue
\mciteSetBstMidEndSepPunct{\mcitedefaultmidpunct}
{\mcitedefaultendpunct}{\mcitedefaultseppunct}\relax
\EndOfBibitem
\bibitem[K{\"o}hler \latin{et~al.}(2021)K{\"o}hler, Kr{\"a}mer, and
  No{\'e}]{kohler2021smooth}
K{\"o}hler,~J.; Kr{\"a}mer,~A.; No{\'e},~F. Smooth Normalizing Flows. Advances
  in Neural Information Processing Systems. 2021\relax
\mciteBstWouldAddEndPuncttrue
\mciteSetBstMidEndSepPunct{\mcitedefaultmidpunct}
{\mcitedefaultendpunct}{\mcitedefaultseppunct}\relax
\EndOfBibitem
\bibitem[Kingma and Welling(2014)Kingma, and Welling]{Kingma2014}
Kingma,~D.~P.; Welling,~M. {Auto-Encoding Variational Bayes}. 2nd International
  Conference on Learning Representations, {ICLR} 2014, Banff, AB, Canada, April
  14-16, 2014, Conference Track Proceedings. 2014\relax
\mciteBstWouldAddEndPuncttrue
\mciteSetBstMidEndSepPunct{\mcitedefaultmidpunct}
{\mcitedefaultendpunct}{\mcitedefaultseppunct}\relax
\EndOfBibitem
\bibitem[Nielsen \latin{et~al.}(2020)Nielsen, Jaini, Hoogeboom, Winther, and
  Welling]{nielsen2020survae}
Nielsen,~D.; Jaini,~P.; Hoogeboom,~E.; Winther,~O.; Welling,~M. Survae flows:
  Surjections to bridge the gap between vaes and flows. Advances in Neural
  Information Processing Systems. 2020\relax
\mciteBstWouldAddEndPuncttrue
\mciteSetBstMidEndSepPunct{\mcitedefaultmidpunct}
{\mcitedefaultendpunct}{\mcitedefaultseppunct}\relax
\EndOfBibitem
\bibitem[Sch\"{u}tt \latin{et~al.}(2018)Sch\"{u}tt, Sauceda, Kindermans,
  Tkatchenko, and M\"{u}ller]{Schutt_JChemPhys2018}
Sch\"{u}tt,~K.~T.; Sauceda,~H.~E.; Kindermans,~P.-J.; Tkatchenko,~A.;
  M\"{u}ller,~K.-R. {SchNet {\textendash} A deep learning architecture for
  molecules and materials}. \emph{J. Chem. Phys.} \textbf{2018}, \emph{148},
  241722\relax
\mciteBstWouldAddEndPuncttrue
\mciteSetBstMidEndSepPunct{\mcitedefaultmidpunct}
{\mcitedefaultendpunct}{\mcitedefaultseppunct}\relax
\EndOfBibitem
\bibitem[Kov{\'{a}}cs \latin{et~al.}(2021)Kov{\'{a}}cs, Oord, Kucera, Allen,
  Cole, Ortner, and Cs{\'{a}}nyi]{Kovacs2021}
Kov{\'{a}}cs,~D.~P.; Oord,~C. V.~D.; Kucera,~J.; Allen,~A.~E.; Cole,~D.~J.;
  Ortner,~C.; Cs{\'{a}}nyi,~G. {Linear Atomic Cluster Expansion Force Fields
  for Organic Molecules: Beyond RMSE}. \emph{J. Chem. Theory Comput.}
  \textbf{2021}, \emph{17}, 7696--7711\relax
\mciteBstWouldAddEndPuncttrue
\mciteSetBstMidEndSepPunct{\mcitedefaultmidpunct}
{\mcitedefaultendpunct}{\mcitedefaultseppunct}\relax
\EndOfBibitem
\bibitem[Thomas \latin{et~al.}(2018)Thomas, Smidt, Kearnes, Yang, Li, Kohlhoff,
  and Riley]{Thomas2018}
Thomas,~N.; Smidt,~T.; Kearnes,~S.; Yang,~L.; Li,~L.; Kohlhoff,~K.; Riley,~P.
  {Tensor field networks: Rotation- and translation-equivariant neural networks
  for 3D point clouds}. \emph{arXiv preprint arXiv:1802.08219} \textbf{2018},
  \relax
\mciteBstWouldAddEndPunctfalse
\mciteSetBstMidEndSepPunct{\mcitedefaultmidpunct}
{}{\mcitedefaultseppunct}\relax
\EndOfBibitem
\bibitem[Klicpera \latin{et~al.}(2020)Klicpera, Gro{\ss}, and
  G{\"u}nnemann]{Klicpera2020}
Klicpera,~J.; Gro{\ss},~J.; G{\"u}nnemann,~S. Directional Message Passing for
  Molecular Graphs. International Conference on Learning Representations
  (ICLR). 2020\relax
\mciteBstWouldAddEndPuncttrue
\mciteSetBstMidEndSepPunct{\mcitedefaultmidpunct}
{\mcitedefaultendpunct}{\mcitedefaultseppunct}\relax
\EndOfBibitem
\bibitem[Batzner \latin{et~al.}(2021)Batzner, Musaelian, Sun, Geiger, Mailoa,
  Kornbluth, Molinari, Smidt, and Kozinsky]{Batzner2021}
Batzner,~S.; Musaelian,~A.; Sun,~L.; Geiger,~M.; Mailoa,~J.~P.; Kornbluth,~M.;
  Molinari,~N.; Smidt,~T.~E.; Kozinsky,~B. {E(3)-Equivariant Graph Neural
  Networks for Data-Efficient and Accurate Interatomic Potentials}. \emph{arXiv
  preprint arXiv:2101.03164} \textbf{2021}, \relax
\mciteBstWouldAddEndPunctfalse
\mciteSetBstMidEndSepPunct{\mcitedefaultmidpunct}
{}{\mcitedefaultseppunct}\relax
\EndOfBibitem
\bibitem[Unke and Meuwly(2019)Unke, and Meuwly]{Unke2019}
Unke,~O.~T.; Meuwly,~M. {PhysNet: A Neural Network for Predicting Energies,
  Forces, Dipole Moments, and Partial Charges}. \emph{J. Chem. Theory Comput.}
  \textbf{2019}, \emph{15}, 3678--3693\relax
\mciteBstWouldAddEndPuncttrue
\mciteSetBstMidEndSepPunct{\mcitedefaultmidpunct}
{\mcitedefaultendpunct}{\mcitedefaultseppunct}\relax
\EndOfBibitem
\bibitem[Smith \latin{et~al.}(2017)Smith, Isayev, and Roitberg]{Smith2017}
Smith,~J.~S.; Isayev,~O.; Roitberg,~A.~E. {ANI-1: an extensible neural network
  potential with DFT accuracy at force field computational cost}. \emph{Chem.
  Sci.} \textbf{2017}, \emph{8}, 3192--3203\relax
\mciteBstWouldAddEndPuncttrue
\mciteSetBstMidEndSepPunct{\mcitedefaultmidpunct}
{\mcitedefaultendpunct}{\mcitedefaultseppunct}\relax
\EndOfBibitem
\bibitem[Chen \latin{et~al.}(2021)Chen, Kr{\"{a}}mer, Charron, Husic, Clementi,
  and No{\'{e}}]{Chen2021}
Chen,~Y.; Kr{\"{a}}mer,~A.; Charron,~N.~E.; Husic,~B.~E.; Clementi,~C.;
  No{\'{e}},~F. {Machine learning implicit solvation for molecular dynamics}.
  \emph{J. Chem. Phys.} \textbf{2021}, \emph{155}, 084101\relax
\mciteBstWouldAddEndPuncttrue
\mciteSetBstMidEndSepPunct{\mcitedefaultmidpunct}
{\mcitedefaultendpunct}{\mcitedefaultseppunct}\relax
\EndOfBibitem
\bibitem[John and Cs{\'{a}}nyi(2017)John, and Cs{\'{a}}nyi]{John2017}
John,~S.~T.; Cs{\'{a}}nyi,~G. Many-Body Coarse-Grained Interactions Using
  Gaussian Approximation Potentials. \emph{The Journal of Physical Chemistry B}
  \textbf{2017}, \emph{121}, 10934--10949\relax
\mciteBstWouldAddEndPuncttrue
\mciteSetBstMidEndSepPunct{\mcitedefaultmidpunct}
{\mcitedefaultendpunct}{\mcitedefaultseppunct}\relax
\EndOfBibitem
\bibitem[{Schr\"odinger, LLC}(2015)]{PyMOL}
{Schr\"odinger, LLC}, The {PyMOL} Molecular Graphics System, Version~1.8.
  2015\relax
\mciteBstWouldAddEndPuncttrue
\mciteSetBstMidEndSepPunct{\mcitedefaultmidpunct}
{\mcitedefaultendpunct}{\mcitedefaultseppunct}\relax
\EndOfBibitem
\end{mcitethebibliography}
\providecommand{\latin}[1]{#1}
\makeatletter
\providecommand{\doi}
  {\begingroup\let\do\@makeother\dospecials
  \catcode`\{=1 \catcode`\}=2 \doi@aux}
\providecommand{\doi@aux}[1]{\endgroup\texttt{#1}}
\makeatother
\providecommand*\mcitethebibliography{\thebibliography}
\csname @ifundefined\endcsname{endmcitethebibliography}
  {\let\endmcitethebibliography\endthebibliography}{}

\end{document}


\section{Methods}
We describe how we can use and evaluate this approach when modeling the CG potential of given protein systems.

\subsection{Models}

As a first step, we have to decide how we choose our CG coordinates. In the second step, we need to design a suitable flow transformation $\Phi(\cdot; \bm\theta_{\mathrm{flow}})$ for the density estimation part. Finally, we need to choose an unconstrained model $V(\bm r; \bm\theta_{\mathrm{pot}})$ that we can train against the forces of the flow-induced potential $\mathcal{V}(\bm r, \bm \eta; \bm\theta_{\mathrm{flow}})$.

\subsubsection{Coarse-graining operator $\bm \Xi$}

Our coarse-graining operator $\bm \Xi \colon \mathbb{R}^{3 N} \rightarrow \mathbb{R}^{3 n}$ projects an $N$-atom peptide conformation onto a subset of $n$ of its backbone atoms. For smaller systems (e.g., capped alanine) we can choose backbone carbons and nitrogens. For larger systems (e.g., fast-folder proteins) we project all-atom conformations onto $C_\alpha$ beads of the backbone (see Fig.~2a in the main text). Other choices of CG mapping that can be described by a linear operator $\bm \Xi$, such as placing a bead on the center of mass for a group of atoms, are also compatible with our methods.

\subsubsection{Flow potential $\mathcal{V}(\bm r, \bm \eta; \bm\theta_{\mathrm{flow}})$}

The inverse flow transformation $\Phi^{-1}(\cdot;\bm\theta_{\mathrm{flow}})$ required to define $\mathcal{V}(\bm r, \bm \eta; \bm\theta_{\mathrm{flow}})$ in Eq.~(8) in the main text is composed by a fixed coordinate transformation into an internal coordinate (IC) representation followed by a trainable normalizing flow (see Fig.~2a).
For the IC transformation, we follow Ref.~\citenum{noe2019boltzmann} and project the Cartesian CG degrees of freedom onto bond lengths, angles, and dihedral torsions of adjacent CG beads.
The normalizing flow follows the architecture in Ref.~\citenum{kohler2021smooth} and consists of coupling layers \cite{dinh2014nice, dinh2016rnvp} where we transform bonds and angles and torsions using either spline or smooth hypertoric transforms \cite{durkan2019neural, kohler2021smooth}. 
We increase expressivity and relax bijectivity constraints of the normalizing flow by introducing $k$ additional latent variables $\bm \eta$ leading to the variational objective Eq.~(7) in the main text.
We choose a factorized base density $q(\bm z_0, \bm z_1) = q(\bm z_0) \cdot q(\bm z_1)$ given by a uniform density $q(\bm z_0) = \mathcal{U}(0, 1)$ for the IC degrees of freedom and 
a isotropic normal density $q(\bm z_1) = \mathcal{N}(0, 1)$ for the latent variables $\bm \eta$.
Finally, we choose the variational density $\tilde \nu(\bm \eta | \bm r) = \tilde \nu(\bm \eta) = \mathcal{N}(0, 1)$ to be an independent isotropic Gaussian, as well. Additional technical details of the implementation of the normalizing flow can be found in Section C.1.

\subsubsection{Unconstrained CG potential $V(\bm r; \bm\theta_{\mathrm{pot}})$}
The unconstrained potential $V(\bm r; \bm\theta_{\mathrm{pot}})$ in the second step follows the \emph{CGnet} architecture in Ref.~\citenum{wang2019machine}. This model transforms Euclidean coordinates of the CG beads into pairwise distances, angles and torsions (similar to the IC transformation introduced in the last paragraph). Then these features are fed into a fully connected neural network to output a scalar energy. Furthermore, the model adds simple repulsive and harmonic \textit{prior energy} terms to this scalar, which prevent steric clashes and bond breaking (see Fig.~2b). We extended the original architecture by introducing skip-connections \cite{he2016deep} and replaced \texttt{tanh}-activations with \texttt{silu}-activations \cite{Elfwing2018SigmoidWeightedLU} which both greatly improved results. Technical details of the implementation of this coarse-graining potential can be found in Section C.2.

\subsection{Training}
The flow potential can be trained on the trajectory data using the variational bound to the likelihood Eq.~(7) in the main text. 
During training, we monitor the 
validation likelihood. We stop training once we observe convergence and pick the checkpoint with the highest validation log-likelihood for later use.

Due to the constraints of a normalizing flow, some of the generated samples from the teacher model are not perfect.
For example, a small number of samples come with significantly larger force magnitudes than the rest of the samples and can disrupt the force-matching training of the CGnet.
We solve this with a rejection sampling scheme and filter flow samples with a set threshold on force magnitude (details in SI-Section D.2). Only few flow samples are removed in this step (Fig.~S4).
Additionally, there are flow samples that contain non-neighboring pairs with shorter distances than the minimum observed in the ground truth, to which lower weights are assigned via a free energy perturbation scheme. In this step a fraction of flow samples is removed that increases with system size, however this fraction is still around 50\% for the largest systems studied here (Fig.~S4). Given the efficiency of flow sampling, even acceptance rates around 1\% would not cause computational bottlenecks.

The unconstrained potential is then trained against the processed flow samples using the variational force-matching objective Eq.~(11) in the main text.
Detailed setup and explanation of the training can be found in Sections D.2 and D.3.

\subsection{CG sampling}
For the flow models, we draw independent samples in latent space according to the prior distributions and use the forward transformation of the flow to map them to CG coordinates.
As for the CGnet-based models, we perform MD simulations in the CG space, similarly as in Refs.~\citenum{wang2019machine} and~\citenum{husic2020coarse}
Except for the time step (5-fs for capped alanine and 2-fs for fast folding proteins), we keep the simulation parameters, such as the thermostat temperature and friction coefficient, consistent with the reference all-atom simulations. Note that there is no simple correspondence between the CG kinetics and thus timescale and the all-atom system~\cite{Nuske_JChemPhys2019}.

In addition, two methods are used to facilitate sampling:~batched simulations from different starting structures and parallel tempering. The starting structures are sampled according to the equilibrium distribution of the all-atom simulations (following Ref.~\citenum{wang2019machine}). The parallel tempering employs two replicas at temperatures 300~K and 450~K for CG alanine and three replicas at temperatures $T_0$, $\sqrt{T_0 \times \text{500~K}}$ and 500~K for fast folders, in which $T_0$ is the simulation temperature used for the reference all-atom data set (see Ref.~\citenum{Lindorff_Science2011}). The conformations at the reference temperature are recorded every 100 or 250 time steps for CG alanine and fast folding proteins, respectively. Other aspects of conducting the simulations can be found in Section D.4.

\section{Proofs and Derivations}

\subsection{Fisher-Identity}
For any sufficiently smooth probability density $p(\bm r, \bm \eta)$ on $\mathbb{R}^{d}$ we have the \textit{Fisher-identity}~\cite{douc2014nonlinear}:
\begin{lemma}
\label{lemma:fisher-identity}
    \begin{align}
        \nabla_{\bm x} \log p(\bm x) = \mathbb{E}_{\bm y \sim p(\bm y|\bm x)}\left[ \nabla_{\bm x} \log p(\bm x|\bm y) \right].
    \end{align}
\end{lemma}
\begin{proof}
\begin{align}
    \nabla_{\bm x} \log p(\bm x) 
                                &= \frac{1}{p(\bm x)} \nabla_{\bm x} p(\bm x) \\
                                &= \frac{1}{p(\bm x)} \nabla_{\bm x} \int d\bm y ~ p(\bm x, \bm y) \\
                                &= \frac{1}{p(\bm x)} \int d\bm y ~ \nabla_{\bm x} p(\bm x, \bm y) \\
                                &= \frac{1}{p(\bm x)} \int d\bm y ~ p(\bm x, \bm y) \nabla_{\bm x} \log p(\bm x, \bm y) \\
                                &= \int d\bm y ~ \frac{p(\bm y, \bm x)}{p(\bm x)} \left[ \nabla_{\bm x} \log p(\bm x|\bm y) + \underbrace{\nabla_{\bm x} \log p(\bm y)}_{=0} \right]\\
                                &= \int d\bm y ~ p(\bm y|\bm x) \nabla_{\bm x} \log p(\bm x|\bm y) \\
                                &= \mathbb{E}_{\bm y \sim p(\bm y|\bm x)}\left[ \nabla_{\bm x} \log p(\bm x|\bm y) \right].
\end{align}
\end{proof}

\subsection{Variational bound on the likelihood of the latent-variable model}

\begin{align}
    \mathcal{L}(\bm\theta_{\mathrm{flow}})
    &= \mathbb{E}_{\bm r \sim \mathcal{D},~\bm \eta \sim \tilde\nu(\bm \eta|\bm r)}\left[-\log p(\bm r, \bm \eta; \bm\theta_{\mathrm{flow}}) \right]\\
    &= -\frac{1}{|\mathcal{D}|}\sum_{\bm r \in \mathcal{D}} \int d\bm\eta ~ \tilde\nu(\bm \eta | \bm r) \log p(\bm r, \bm \eta; \bm\theta_{\mathrm{flow}}) \\
    &= -\frac{1}{|\mathcal{D}|}\sum_{\bm r \in \mathcal{D}} \int d\bm\eta ~ \tilde\nu( \bm \eta | \bm r) (\log p(\bm r; \bm\theta_{\mathrm{flow}}) + \log p(\bm \eta | \bm r; \bm\theta_{\mathrm{flow}})) \\
    &= -\frac{1}{|\mathcal{D}|}\sum_{\bm r \in \mathcal{D}} \int d\bm\eta ~  \tilde\nu( \bm \eta | \bm r) \log p(\bm r; \bm\theta_{\mathrm{flow}}) 
      -\frac{1}{|\mathcal{D}|}\sum_{\bm r \in \mathcal{D}} \int d\bm\eta ~ \tilde\nu( \bm \eta | \bm r)  \log p(\bm \eta | \bm r; \bm\theta_{\mathrm{flow}}) \\
      &= -\frac{1}{|\mathcal{D}|}\sum_{\bm r \in \mathcal{D}}  \log p(\bm r; \bm\theta_{\mathrm{flow}}) +
     \mathbb{E}_{\bm r \sim \mathcal{D}}\left[H\left[\tilde\nu(\cdot|\bm r), p(\cdot|\bm r;\bm\theta_{\mathrm{flow}})\right]\right] \\
     &= \mathbb{E}_{\bm r \sim \mathcal{D}}\left[-\log p(\bm r; \bm\theta_{\mathrm{flow}}) \right] + \mathbb{E}_{\bm r \sim \mathcal{D}}\left[\underbrace{H\left[\tilde\nu(\cdot|\bm r), p(\cdot|\bm r;\bm\theta_{\mathrm{flow}})\right]}_{\geq 0}\right]\\
     &\geq \mathbb{E}_{\bm r \sim \mathcal{D}}\left[-\log p(\bm r; \bm\theta_{\mathrm{flow}}) \right]
\end{align}
where $H(\cdot, \cdot)$ denotes the cross-entropy.

\subsection{Consistency of teacher-student force-matching}
We first note that
\begin{align}
    \bm \tilde{f}(\bm r, \bm \eta; \bm\theta_{\mathrm{flow}}) 
    &= \nabla_{\bm r} \log p(\bm r, \bm \eta; \bm\theta_{\mathrm{flow}}) \\
    &= \nabla_{\bm r} \log p(\bm r | \bm \eta; \bm\theta_{\mathrm{flow}}) + \nabla_{\bm r} \log p(\bm \eta; \bm\theta_{\mathrm{flow}})\\
    &= \nabla_{\bm r} \log p(\bm r | \bm \eta; \bm\theta_{\mathrm{flow}}).
\end{align}
Combining Jensen's inequality with Lemma~\ref{lemma:fisher-identity}, we obtain
\begin{align}
    \mathcal{L}(\bm\theta_{\mathrm{pot}}) 
    &= \mathbb{E}_{(\bm r, \bm \eta) \sim p(\bm r, \bm \eta; \bm\theta_{\mathrm{flow}})}\left[\left\| \nabla_{\bm r} V(\bm r; \bm\theta_{\mathrm{pot}}) + \bm \tilde{f}(\bm r, \bm \eta; \bm\theta_{\mathrm{flow}})\right\|_2^2\right] \\
    &= \mathbb{E}_{(\bm r, \bm \eta) \sim p(\bm r, \bm \eta; \bm\theta_{\mathrm{flow}})}\left[\left\| \nabla_{\bm r} V(\bm r; \bm\theta_{\mathrm{pot}}) + \nabla_{\bm r} \log p(\bm r | \bm \eta; \bm\theta_{\mathrm{flow}}))\right\|_2^2\right] \\
    &= \mathbb{E}_{\bm r \sim p(\bm r; \bm\theta_{\mathrm{flow}})}\left[\mathbb{E}_{\bm \eta \sim p(\bm \eta | \bm r; \bm\theta_{\mathrm{flow}})}\left[\left\| \nabla_{\bm r} V(\bm r; \bm\theta_{\mathrm{pot}}) + \nabla_{\bm r} \log p(\bm r | \bm \eta; \bm\theta_{\mathrm{flow}}))\right\|_2^2\right]\right] \\
    &\geq \mathbb{E}_{\bm r \sim p(\bm r; \bm\theta_{\mathrm{flow}})}
    \left[ \left\|  
        \nabla_{\bm r} V(\bm r; \bm\theta_{\mathrm{pot}}) + \mathbb{E}_{\bm \eta \sim p(\bm \eta | \bm r; \bm\theta_{\mathrm{flow}})}\left[ \nabla_{\bm r} \log p(\bm r | \bm \eta; \bm\theta_{\mathrm{flow}}))
    \right] \right\|_2^2 \right]\\
    &= \mathbb{E}_{\bm r \sim p(\bm r; \bm\theta_{\mathrm{flow}})}
    \left[ \left\|  
        \nabla_{\bm r} V(\bm r; \bm\theta_{\mathrm{pot}}) + \nabla_{\bm r} \log p(\bm r; \bm\theta_{\mathrm{flow}}))
     \right\|_2^2 \right]\\
     &= \mathbb{E}_{\bm r \sim p(\bm r; \bm\theta_{\mathrm{flow}})}
    \left[ \left\|  
        \nabla_{\bm r} V(\bm r; \bm\theta_{\mathrm{pot}}) + \bm f(\bm r; \bm\theta_{\mathrm{flow}}))
     \right\|_2^2 \right].
\end{align}
Furthermore, we have 
\begin{align}
    \mathbb{E}_{(\bm r, \bm \eta) \sim p(\bm r, \bm \eta; \bm\theta_{\mathrm{flow}})}\left[ \nabla_{\bm r} V(\bm r; \bm\theta_{\mathrm{pot}})^T \bm \tilde f(\bm r, \bm \eta; \bm\theta_{\mathrm{flow}}) \right]
    &= 
    \mathbb{E}_{\bm r \sim p(\bm r; \bm\theta_{\mathrm{flow}})}\left[ \nabla_{\bm r} V(\bm r; \bm\theta_{\mathrm{pot}})^T \bm f(\bm r; \bm\theta_{\mathrm{flow}}) \right].
\end{align}
From which we can derive
\begin{align}
    & \mathbb{E}_{(\bm r, \bm \eta) \sim p(\bm r, \bm \eta; \bm\theta_{\mathrm{flow}})}\left[\left\| \nabla_{\bm r} V(\bm r; \bm\theta_{\mathrm{pot}}) + \bm \tilde f(\bm r, \bm \eta; \bm\theta_{\mathrm{flow}})\right\|_2^2\right] - \mathbb{E}_{\bm r \sim p(\bm r; \bm\theta_{\mathrm{flow}})}
    \left[ \left\|  
        \nabla_{\bm r} V(\bm r; \bm\theta_{\mathrm{pot}}) + \bm f(\bm r; \bm\theta_{\mathrm{flow}}))
     \right\|_2^2 \right]\\
     &= \mathbb{E}_{(\bm r, \bm \eta) \sim p(\bm r, \bm \eta; \bm\theta_{\mathrm{flow}})}\left[
        \|\bm \tilde f(\bm r, \bm \eta; \bm\theta_{\mathrm{flow}})\|_2^2 - \|\bm f(\bm r; \bm\theta_{\mathrm{flow}})\|_2^2
     \right].
\end{align}
Thus, the variational gap introduced by the latent variables does not depend on $\bm\theta_{\mathrm{pot}}$ which makes 
\begin{align}
    \mathbb{E}_{(\bm r, \bm \eta) \sim p(\bm r, \bm \eta; \bm\theta_{\mathrm{flow}})}\left[\nabla_{\bm\theta_{\mathrm{pot}}} \left\| \nabla_{\bm r} V(\bm r; \bm\theta_{\mathrm{pot}}) + \bm \tilde f(\bm r, \bm \eta; \bm\theta_{\mathrm{flow}})\right\|_2^2\right]
\end{align}
an unbiased gradient estimator of the force matching loss.

\subsection{Reweighting of the flow samples to incorporate a pairwise repulsion}
Assuming that the CG potential implied by a trained flow model $\mathcal{V}(\bm r, \bm \eta; \bm\theta_{\mathrm{flow}})$ lacks a repulsion term $U_\text{repul}$, we aim for obtaining samples that follows the Boltzmann distribution according to a more physical CG potential $\mathcal{V}(\bm r, \bm \eta; \bm\theta_{\mathrm{flow}}) + U_\text{repul}$ by free-energy perturbation~\cite{zwanzig1954high}.
For convenience, we let $U_0:=\mathcal{V}(\bm r, \bm \eta; \bm\theta_{\mathrm{flow}})$ and $\Delta U:=U_\text{repul}$ and omit the dependency on augmented variable $\bm\eta$ for the derivation below.

The possibility for observing a certain CG conformation $\bm r$, i.e., the Boltzmann weight according to a potential $U_0$, is
\begin{equation}
    p_0(\bm r) = \frac{1}{Z_0} e^{-\beta U_0(\bm r)},\;Z_0=\int_{\Gamma_0} e^{-\beta U_0(\bm r')} \mathop{}\!\mathrm{d}\bm r'.
\end{equation}
Similarly, we have the Boltzmann weight for $U_0 + \Delta U$:
\begin{equation}
    p_1(\bm r) = \frac{1}{Z_1} e^{-\beta \left[U_0(\bm r) + \Delta U(\bm r)\right]},\;Z_1=\int_{\Gamma_1} e^{-\beta \left[U_0(\bm r') + \Delta U(\bm r')\right]} \mathop{}\!\mathrm{d}\bm r'.
\end{equation}
When the $\Delta U$ is finite, the phase spaces $\Gamma_0$ and $\Gamma_1$ are identical, and we have:
\begin{align}
Z_1
&= \int_{\Gamma_0} e^{-\beta \left[U_0(\bm r') + \Delta U(\bm r')\right]} \mathop{}\!\mathrm{d}\bm r' \\
&= \int_{\Gamma_0} e^{-\beta U_0(\bm r')} \cdot e^{-\beta \Delta U(\bm r')} \mathop{}\!\mathrm{d}\bm r' \\
&= Z_0\int_{\Gamma_0} \frac{e^{-\beta U_0(\bm r')}}{Z_0} \cdot e^{-\beta \Delta U(\bm r')} \mathop{}\!\mathrm{d}\bm r' \\
&=Z_0\int_{\Gamma_0} p_0(\bm r') \cdot e^{-\beta \Delta U(\bm r')} \mathop{}\!\mathrm{d}\bm r'. \label{eq:z-ratio}
\end{align}
(Note that even when the repulsion energy goes to infinity on a zero-measure set where two particles overlap $\left\{\bm r|\exists j,k\in \text{particles},\;s.t.~\vec r_j=\vec r_k\right\}$, e.g., for a $1/d^p$ or LJ-like repulsion, the above equality still holds.)
For convenience, we introduce the thermodynamic averaging operator 
\begin{equation}
    \left\langle \phi(\bm r) \right\rangle_0 := \int_{\Gamma_0} p_0(\bm r') \phi(\bm r') \mathop{}\!\mathrm{d}\bm r',
\end{equation}
and then \eqref{eq:z-ratio} can be rewritten into the following:
\begin{equation}
    \frac{Z_0}{Z_1} = \frac{1}{\left\langle e^{-\beta \Delta U(\bm r)} \right\rangle_0}.
\end{equation}
Subsequently, we can evaluate $p_1$ if we know $p_0$ for any given conformation $\bm r$:
\begin{align}
    p_1(\bm r) &= Z_1^{-1} e^{-\beta U_0(\bm r)} e^{-\beta \Delta U(\bm r)} \\
    &= Z_1^{-1} \left[ Z_0p_0(\bm r) \right] e^{-\beta \Delta U(\bm r)} \\
    &= p_0(\bm r) \cdot \left[ \frac{Z_0}{Z_1} \cdot e^{-\beta \Delta U(\bm r)} \right] \\
    &= p_0(\bm r) \cdot \underline{ \frac{e^{-\beta \Delta U(\bm r)}}{\left\langle e^{-\beta \Delta U(\bm r)} \right\rangle_0} }.
\end{align}
The underlined expression is the reweighting factor that connects the original potential $U_0$ and the perturbed potential $U_0+\Delta U$.
Given a set of coordinates $\left\{ \bm r_i \right\}_N$ sampled from the flow (i.e., $U_0$), we can approximate the reweighting factor by:
\begin{align}
    w_i &:= \frac{e^{-\beta \Delta U(\bm r_i)}}{\left\langle e^{-\beta \Delta U(\bm r)} \right\rangle_0} \\
        &\approx \frac{e^{-\beta \Delta U(\bm r_i)}}{\left[\sum_l e^{-\beta \Delta U(\bm r_l)}\right]/N}. \label{eq:weight_expression}
\end{align}
This factor can be used for training the secondary CGnet model, including computing the marginal mean and standard deviations of the bond and angle dimensions (for prior energy parameters) as well as computing the force matching loss.
As an example, the reweighted force matching loss (cf.~Eq.~11 in maintext) over a set of flow-samples $\left\{ (\bm r_i, \tilde {\bm f_i}) \right\}_N$ becomes:
\begin{equation}
    \mathcal{L}(\bm\theta_{\mathrm{pot}}) \approx \sum_{i=1}^N \frac{w_i}{N} \left\| \tilde {\bm f'_i} + \nabla_{\bm r} V(\bm r; \bm\theta_{\mathrm{pot}})  \right\|^2_2,
\end{equation}
in which the weights $w_i$ is evaluated via Eq.~\ref{eq:weight_expression} over the whole set of flow samples and the new force corresponds to the modified potential:
\begin{equation}
    \tilde {\bm f'_i}(\bm r_i, \bm\eta_i) 
    = -\nabla_{\bm r}\left[ U_0 + \Delta U \right] (\bm r_i) 
    = \tilde {\bm f_i} -\nabla_{\bm r}\left[ \Delta U \right] (\bm r_i) 
    = \tilde {\bm f_i} + \bm f_\text{repul} (\bm r_i).
\end{equation}

\section{Additional details on the models}

\subsection{Normalizing flow architecture}
The normalizing flow architecture is sketched in Fig.~\ref{fig:normalizing-flow-flow-chart}. After transforming euclidean coordinates ($\mathbf{r}$) into internal coordinates (Fig.~\ref{fig:normalizing-flow-flow-chart}, \texttt{ic trafo})~\cite{noe2019boltzmann}, we apply an inverse CDF transform (Fig.~\ref{fig:normalizing-flow-flow-chart}, \texttt{iCDF}) onto bonds and angles, such that they are mapped into the unit interval. For the inverse CDF, we assume bonds and angles to follow a truncated Gaussian distribution where the parameters (truncation bounds, mean, variance) are estimated from the data (Fig.~\ref{fig:normalizing-flow-flow-chart}, \texttt{summary statistics}). These whitened ICs are then transformed together with the latent variables $\bm \eta$ into uniform and Gaussian densities, respectively, using the trainable flow, which is composed of three sub-flows transforming torsions, angles and bonds, respectively (Fig.~\ref{fig:normalizing-flow-flow-chart}, \texttt{torsion flow}, \texttt{angle flow}, \texttt{bond flow}). All three sub-flows consist of an alternating stack of coupling~\cite{dinh2014nice, dinh2016rnvp} and shuffling blocks. The coupling blocks transform one group of variables (e.g., one group of torsions) conditioned on a context (e.g., the other group of torsions and the latent variables). The transforms (Fig.~\ref{fig:normalizing-flow-flow-chart}, \texttt{smooth}, \texttt{spline} or \texttt{affine}) themselves require parameters, which are computed as the output of a trainable conditioner neural network (Fig.~\ref{fig:normalizing-flow-flow-chart}, \texttt{NN}) with the context as input.
The transformed variables have different domains. While bonds and angles are supported on the unit interval, the torsions are supported on a circle, and the latent variables on a real vector space. To satisfy topological constraints \cite{kohler2021smooth, rezende2020normalizing, wirnsberger2020targeted}, the following transforms are used:~for the latent variables we use simple affine transforms~\cite{dinh2016rnvp} (Fig.~\ref{fig:normalizing-flow-flow-chart}, \texttt{C2}). For bonds we rely on spline transforms~\cite{durkan2019neural} (Fig.~\ref{fig:normalizing-flow-flow-chart}, \texttt{C4}), while for angles and torsions we use smooth transforms~\cite{kohler2021smooth} (Fig.~\ref{fig:normalizing-flow-flow-chart}, \texttt{C1}, \texttt{C3}).
The conditioner neural network is a simple two layer dense net (Fig.~\ref{fig:normalizing-flow-flow-chart},  \texttt{NN}). 
If torsions are part of the context, we satisfy the periodic boundary condition by a projection onto a sin/cos basis before feeding them into the network (Fig.~\ref{fig:normalizing-flow-flow-chart}, \texttt{C1}, \texttt{C2}, \texttt{C3}).

Although the transformation of a certain feature in each coupling layer can be conditioned upon all other features (i.e., they serve as input to the \texttt{NN}), we generally used the following restricted version in our experiments. It is based on an assumption of hierarchical dependency among the internal coordinates and helps to reduce computational overhead.
Here we walk through the construction details of all layers in reverse order, i.e., from prior distribution to the actual IC distribution. Note that each feature channel can also be divided in halves to encourage mixing.
\begin{enumerate}
    \item Prior distributions, including the following channels
    \begin{enumerate}
        \item T (torsional), A (angular), B (bond-length channels): uniform distribution
        \item AUG (latent/augmentation channels): normal distribution
    \end{enumerate}

    \item Torsion flow (T1, T2 are two halves of an equal split of all torsion channels), consisting of two or four torsion blocks, each defined as
    \begin{enumerate}
        \item AUG $\leftarrow\text{coupling flow (affine)}-$ (T1, T2)
        \item T1 $\leftarrow\text{coupling flow (smooth)}-$ (T2, AUG)
        \item T2 $\leftarrow\text{coupling flow (smooth)}-$ (T1, AUG)
    \end{enumerate}

    \item Angle flow (A1, A2 are two halves of an equal split of all angle channels A), consisting of two angle blocks, each defined as
    \begin{enumerate}
        \item A1 $\leftarrow\text{coupling flow (smooth)}-$ (A2, T)
        \item A2 $\leftarrow\text{coupling flow (smooth)}-$ (A1, T)
    \end{enumerate}

    \item Bond flow, consisting of one bond block defined as
    \begin{enumerate}
        \item B $\leftarrow\text{coupling flow (spline)}-$ (A, T)
    \end{enumerate}

    \item IC handling, including
    \begin{enumerate}
        \item Inverse CDF transforms on B and A (to a truncated normal distribution)
        \item Inverse IC transformation
    \end{enumerate}
\end{enumerate}
The hyper-parameters for the experiments are given in Table~\ref{tab:hyper-params-flow}.
 
\begin{figure}
\centering
\includegraphics[width=\textwidth]{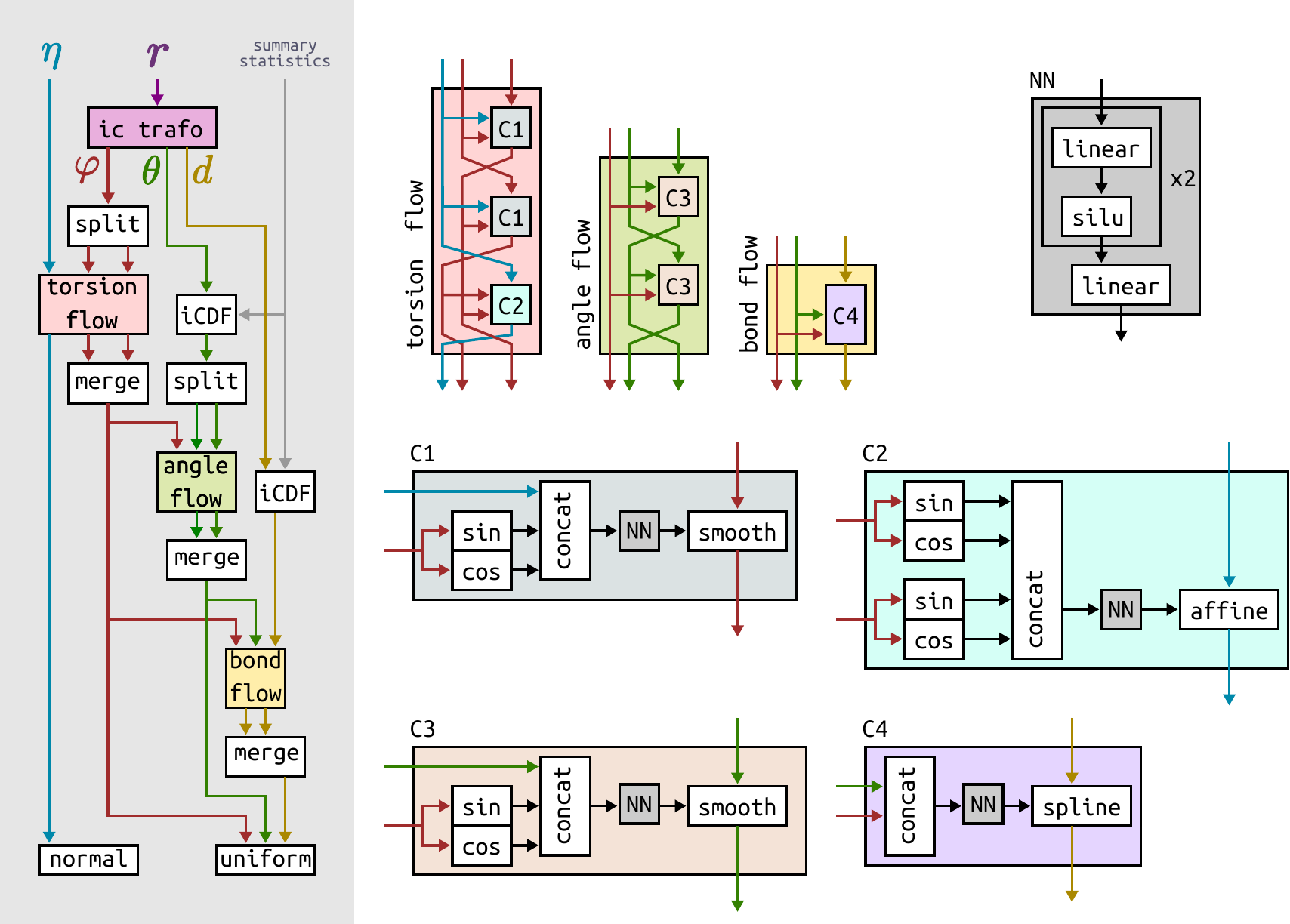}
\caption{Architecture of the normalizing flow.}
\label{fig:normalizing-flow-flow-chart}
\end{figure}

\begin{table}\centering
\caption{Flow-specific hyper-parameters used in the experiments}
\label{tab:hyper-params-flow}
\begin{tabular}{lrrrr}
System & Hidden units for \texttt{NN} & No.~of torsion blocks & No.~of latent dimensions \\
\midrule
Capped-Alanine & 128, 1024, 128 & 2 & 2\\
Fast-Folders & 128, 1024, 128 & 4 & 2\\
\bottomrule
\end{tabular}
\end{table}

\subsection{CGnet architecture}
The CGnet architecture~\cite{wang2019machine} is used as the ``student'' model that distills coarse graining knowledge from the flow samples. Essentially, all the internal coordinates (i.e., bond lengths, angles and dihedral angles) as well as the pairwise distances between nonbonded bead pairs are computed for input conformations. After a Z-score layer for whitening, they are fed into a multilayer perceptron (MLP) with hyperbolic tangent (\texttt{tanh}) activation function of a fixed width. The output energy is obtained as a weighted sum of the MLP output and its negative gradient with respect to the input coordinates gives forces. The force matching error can be computed between the neural network prediction and the mapped CG force from all-atom reference and forms the loss function for training.

For capped alanine, we mostly followed the architectural choices of the original publication~\cite{wang2019machine}. The only difference we introduced was adding skip connections between the output of each layer except for the input and output layers. We found that such changes reduced number of the epochs necessary for training convergence, and the accuracy of resulted models are comparable with the reported behavior in Ref.~\citenum{wang2019machine} when trained on the ground truth forces from all-atom simulations. The choice for prior energy terms was kept intact:~harmonic potential terms were exerted on the bond and angle features, whose parameters were based on the ground-truth statistics. For a fair comparison, the same set of hyperparameters were used for both the conventional CGnet and the Flow-CGnet models.

For coarse graining of fast folding proteins, we found it necessary to introduce a few changes, such that the CGnet could correctly learn the free energy landscape from the flow models. These changes include relaxing the Lipschitz regularization strength, changing the activation function to Sigmoid Linear Unit (SiLU) as well as increasing the number of CGnet layers. For chignolin increasing the number of layers led to overfitting. Therefore, we stayed with 5 layers for this special case. 
Similar to Wang et al.'s experiments on CG chignolin, a repulsion term between the nonbonded bead pairs proved to be necessary for maintaining a resonable exclusion volume and excluding unphysical crashes~\cite{wang2019machine}. We found the numerical stability and the accuracy of estimated free energy surface from simulation were sensitive to the choice of function form as well as parameters for the repulsion. The final choice of repulsion turned out to be a $C^\infty$ piece-wise function:
\begin{equation}
    u_\text{repul}(\vec r_i, \vec r_j) = 
    \begin{cases}
        1600k_BT\cdot\left( d_{ij} - \sigma_{ij} \right)^2 & \text{if } j>i+1 \text{ and } d_{ij}<\sigma_{ij} \\
        0 & \text{otherwise}
    \end{cases},
\end{equation}
where $d_{ij}:=\left| \vec r_i - \vec r_j \right|$ and the endpoint $\sigma_{ij}=0.36\,\text{nm}$ for glycine-involved pairs and $0.42\,\text{nm}$ for all other pairs.
The potential is computed for every non-neighboring pairs in the CG molecule and the results are summed to give the repulsive prior term. Note that the same repulsion function was also used for reweighting flow samples, which is based on SI-Section B.4.

A list of concrete hyperparamter choices can be found in Table~\ref{tab:hyper-params-cgnet}.

\begin{table}\centering
\caption{Hyperparameters of CGnets in the experiments}
\label{tab:hyper-params-cgnet}
\begin{tabular}{l
>{\raggedleft\arraybackslash}p{0.2\textwidth}
>{\raggedleft\arraybackslash}p{0.1\textwidth}
>{\raggedleft\arraybackslash}p{0.1\textwidth}
>{\raggedleft\arraybackslash}p{0.12\textwidth}
>{\raggedleft\arraybackslash}p{0.2\textwidth}}
System & No.~of fully connected layers & Neurons in each layer & Activation function & Lipschitz regularization strength & Prior terms\\
\midrule
Capped-Alanine & 5 & 160 & \texttt{tanh} & 4 & Harmonic terms on bonds and angles\\
Fast-Folders & 5 (chignolin) / 8 (others) & 160 & \texttt{SiLU} & 10 & Harmonic terms on bonds and angles + repulsion\\
\bottomrule
\end{tabular}
\end{table}

\section{Additional details on the experiments}
\label{sec:details}
\subsection{All-atom simulation for capped alanine}
The training set for capped alanine was generated in house with conventional all-atom MD simulations in OpenMM~\cite{Eastman_PLoSComputBiol2017}. The simulation system was set up according to Ref.~\citenum{wang2019machine}. %
After equilibration at target temperature 300K for 10ns, the peptide coordinates and forces were recorded with a 2-ps interval. We performed four independent simulation runs starting from different initial structures. Each run is of length 500~ns, resulting in four times 250,000 sample points.

\subsection{Flow training, sample generation and post-processing}
For capped alanine, the batch size for flow model training was set to 256. ADAM optimizer~\cite{kingma2014adam} was used and the learning rate was set to $0.001$. We performed a four-fold cross validation by picking each trajectory in turn for validation, while using samples from the remaining three for training. In order to evaluate the scaling behavior of models with respect to the numbers of available training samples, models were obtained from training sets with different sizes, ranging from the maximum available $750,000$ down to $10,000$. The subsampling was done with random sampling.
Training was performed for different number of epochs with respect to the training set sizes:
\begin{itemize}
    \item 750,000 and 500,000: 30 epochs,
    \item 200,000: 75 epochs,
    \item others: 100 epochs.
\end{itemize}
For the fast folding proteins, the batch size was 128 and the maximum epoch number was uniformly 50.

The negative log likelihood of the flow on validation set is computed after each training epoch. Convergence in the validation loss was observed for all cases. The set of weights corresponding to the lowest validation loss was saved as checkpoint and 1,048,576 samples were generated with the best model and forces are calculated accordingly.

For capped alanine, we chose to discard samples whose corresponding forces exceed $\sqrt{1.5\times10^5}k_BT/\text{nm}$ in magnitude, which is defined as $\sqrt{\frac{1}{N}\sum_i\lVert \mathbf f \rVert_2^2}$. For fast folding proteins, we filtered the samples with an upper limit for force magnitude of $\sqrt{8\times10^4}k_BT/\text{nm}$, and then computed the weights for the remaining samples according to the repulsion term introduced in SI-Section C.2 and expression in B.4. The weights were used to compute the statistics for defining the harmonic prior terms for Flow-CGnets and the weighted force-matching error for training and validation. A summary of the effect of both post-processing steps on the final effective train set size can be found in Section S.7.

\subsection{Flow-CGnet training}
The CGnet training for capped alanine used ADAM optimizer as well. The batch size was 128 and initial learning rate was $0.003$. An exponential decay was applied on the learning rate every 5 epochs, such that the target learning rate $10^{-5}$ was reached in 50 epochs. The incoming flow samples were randomly shuffled and then divided according 80\%--20\% ratio for training and validation sets. Model checkpoints were saved every two epochs at the epochal end, from which the one with minimum force matching loss on the validation set was used for later simulation.
The training of Flow-CGnet for fast folders followed essentially the same setup, expect that the learning rate decayed every 15 epochs over 75 training epochs in total, and the loss calculation was based on repulsion-reweighting.

\subsection{CGnet simulations}
The friction constant for CG Langevin dynamics is set to 1~$\text{ps}^{-1}$. 
For capped alanine, we performed 2-ns parallel tempering (PT) simulations for both Flow-CGnet and conventional CGnet models.
For fast folding proteins, we performed both 50-ns PT and 50-ns normal Langevin simulations. 
The exchange-proposing interval for PT simulations is 2~ps for capped alanine and 5~ps for fast folding proteins.
In simulations, 100 independent trajectories were obtained in parallel for each molecule, so as for accumulating sufficient samples with reduced computational overhead for each time step. 

\subsection{tICA coordinates}
For a low-dimensional comparison between coarse-grained and atomistic models, we performed time-lagged independent component analysis (tICA) \cite{Naritomi2011,Schwantes_JChemTheoryComput2013,Perez_JChemPhys2013} using deeptime \cite{Hoffmann2022deeptime}. As features, we computed all pairwise distances between C$_\alpha$ atoms as well as all dihedral angles between every four consecutive C$_\alpha$ atoms in a protein chain. Time-lagged feature covariances were computed using a lagtime of 20 ns (100 frames) for each set of trajectories from \cite{Lindorff_Science2011}. This choice of lag time enables recognizing all mean transition path times, which range from 40 to 700 ns for the considered trajectories. The data was projected on the two slowest modes (corresponding to the highest eigenvalues) for plotting and further analysis.

\subsection{Scaling of the flow and Flow-CGnet modeling accuracy versus the amount of samples used for training}

During the two-stage model training of Flow-CGnets, there are two factors that contribute to the data efficiency:~the number of all-atom conformations used as a reference to train the flow model, and the number of CG samples generated by the flow for training the secondary CGnet model.
These two factors also directly affect the training time as well as computational complexity.
This section is intended to provide some insights to distinguish the importance of these two factors and help to identify the real advantage of the Flow-CGnet over conventional CGnet approach.

For the sake of clarity, we will use $N_1$ to represent the number of samples (coordinates only) for training the flow, which are ultimately from the all-atom MD simulation. $N_2$ represents the number of flow samples (coordinates and forces) generated and used in the subsequent training of Flow-CGnet. For the conventional CGnet, we use $N_0$ to represent the number of CG-mapped all-atom coordinates and forces for the training process.
As for the accuracy of the trained models, we use the same two criteria as we used in the main text for capped alanine example:~the KL divergence between the equilibrium distributions of the CG model (\texttt{KL}) and mapped all-atom references, as well as the mean-squared-error of free energy (\texttt{MSE}).
Both criteria are evaluated on representative 2D reaction coordinates that can distinguish all metastable states in the all-atom simulation.
In general, the same reaction coordinates are chosen as in the main text.
For capped alanine, it is the Ramachadran plot, i.e., the dihedral angle $\phi$ and $\psi$.
For chignolin, we use the first two TICs from TICA analysis over CG-mapped all-atom trajectories.
For each choice of $N_0$ as well as each combination of $N_1$ and $N_2$, we evaluate~\texttt{KL} and \texttt{MSE} over four models obtained from a four-fold cross validation experiment and report the mean and (sample) standard deviation.

\begin{figure}
    \centering
    \includegraphics[width=\textwidth]{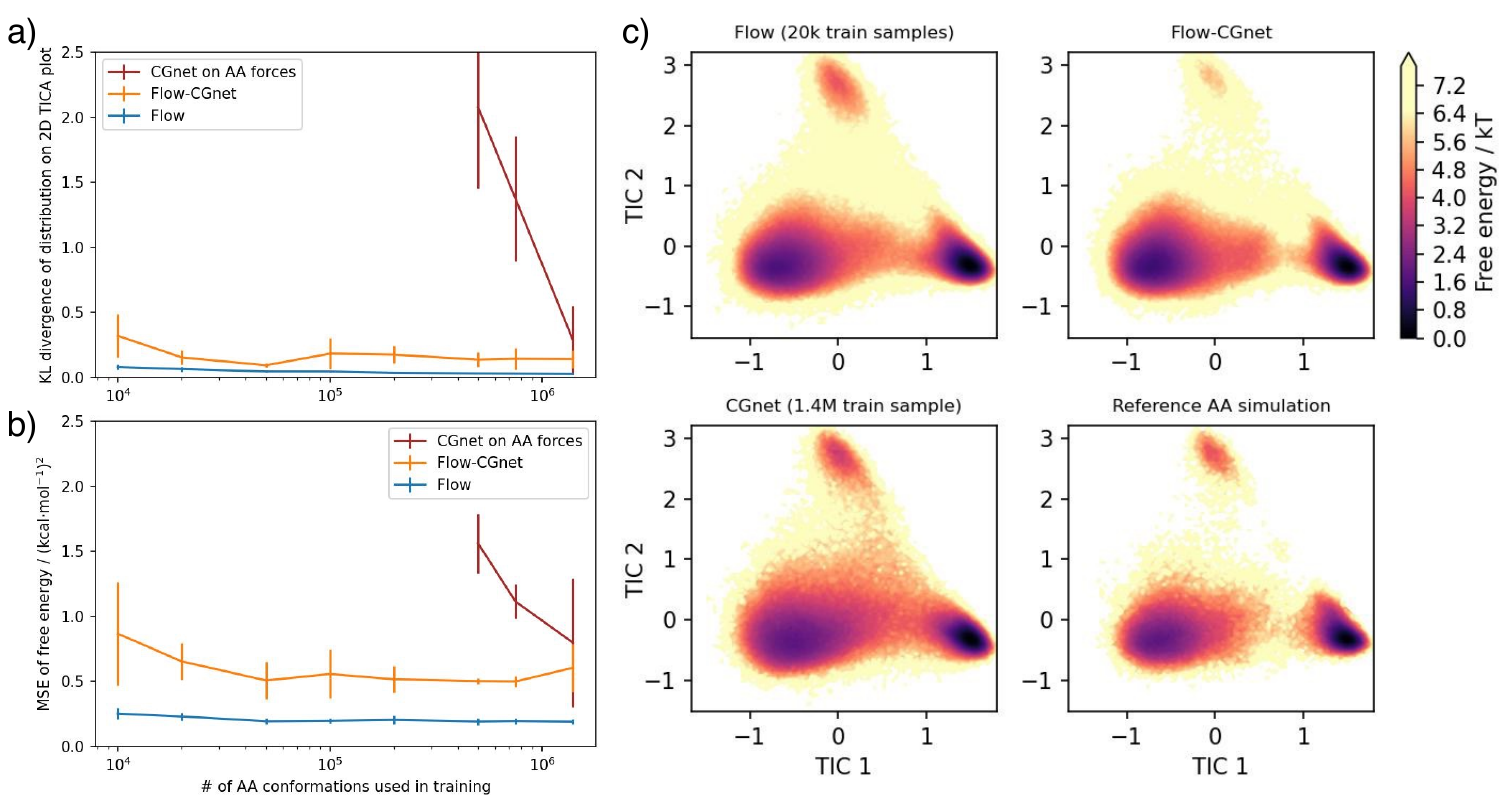}
    \caption{
    Data efficiency of coarse-grained Chignolin that was simulated with CHARMM22* on custom hardware. Note that in contrast to the simulations on the Anton supercomputer, these simulations explore a third (misfolded) state (Panel c, top). 
    \textbf{a)} KL divergence and \textbf{b)} mean square error between coarse-grained and all-atom free energies on the first two time-lagged independent components (TICs). The errors for CGnet trained directly with force matching (brown) are only shown for amounts of data for which the training process successfully converged.
    \textbf{b)}
    Free energy as a function of the first two TICs, learned by the Flow from 20,000 training samples (top left), the Flow-CGnet trained by this flow (top right), CGnet trained directly with force matching using the full dataset ($1.4 \cdot 10^6$ samples, bottom left), and the reference all-atom data (bottom right). The pure flow free energy outperforms CGnet using any tested amount of samples, Flow-CGnet is on par with CGnet when only trained on a factor of 70 less data.
    }
    \label{fig:cln-scale}
\end{figure}

In Fig.~\ref{fig:cln-scale}, we compare the scaling behavior of flow and Flow-CGnet models with conventional CGnets over the smallest fast folding protein chignolin.
Because the original DESRES data does not directly comprise forces and the water coordinates necessary for evaluating the forces are missing, we used a home-brew data set from all-atom MD simulation with a similar force field and simulation setup~\cite{wang2019machine, husic2020coarse}. 
The Flow-CGnet architectures and hyperparameters stay the same as the one used for chignolin in the main text.
Here we report the scaling of~\texttt{KL} and \texttt{MSE} according to different $N_0$s for conventional CGnet and $N_1$s for flow and Flow-CGnet (Fig.~\ref{fig:cln-scale}a and b).
As the curves indicate, the flow as well as Flow-CGnet can already reach rather good accuracy when exposed to very limited amount of reference conformations, while for traditional force-matching it takes almost the whole data set to reach a similar performance.
This can be validated by the 2D free energy surface in Fig.~\ref{fig:cln-scale}c. For the flow and Flow-CGnet, distributions are plotted according to equilibrium samples, where the flow was trained over 20,000 conformations mapped from all-atom.
For comparison, there are also results from full-power conventional CGnet as well as the reference surface.

\begin{figure}
    \centering
    \includegraphics[width=\textwidth]{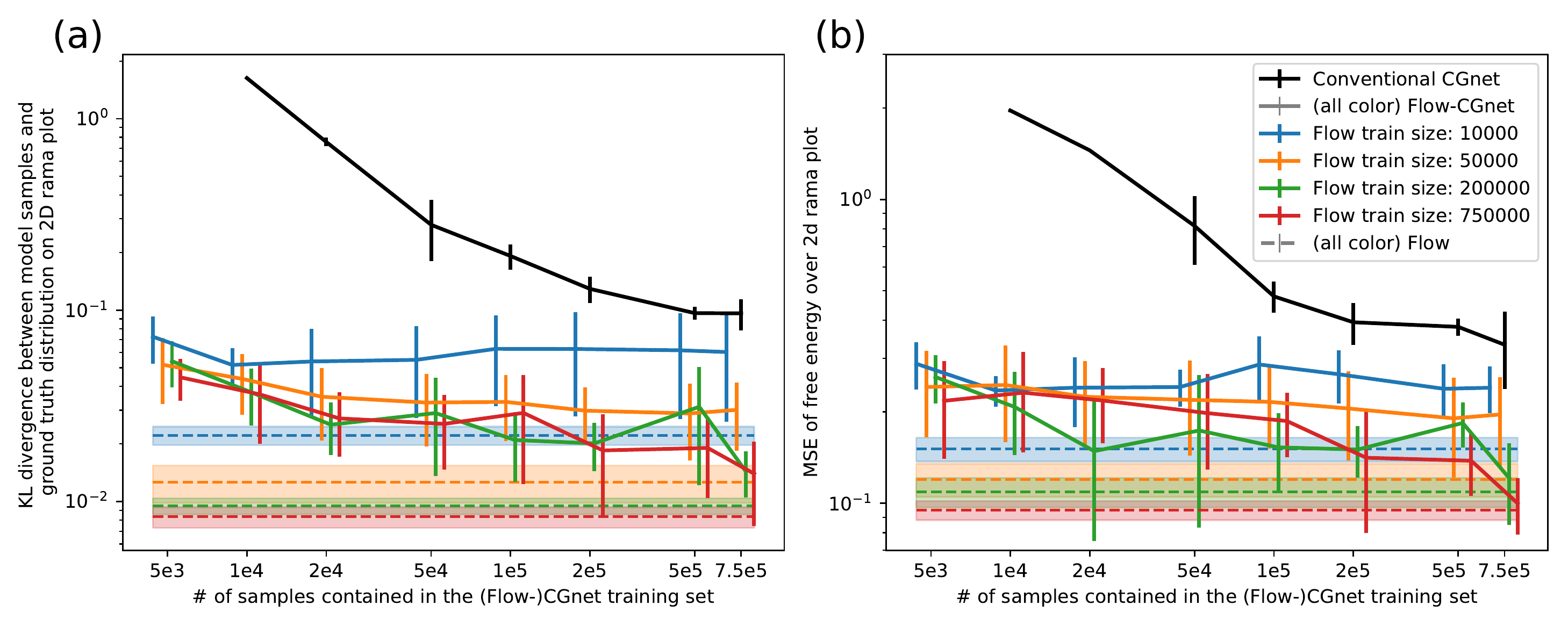}
    \caption{Accuracy of flow, Flow-CGnet as well as conventional CGnet in modeling CG capped alanine as a function of number of training samples fed to the CGnet or Flow-CGnet models. \textbf{a)} the KL divergence of distributions on the 2D Ramachandran plots of the model samples and the ground truth all-atom samples; \textbf{b)} the mean-squared-error of the 2D free energy. The black solid line shows the performance of conventional CGnets, while the colored solid curves for the Flow-CGnets. As a reference lower bound for both criteria, we include the performance of the corresponding flow for generating samples (colored dashed lines). Each value for plotting is averaged over a four-fold cross validation run, and the error bar as well as filled regions indicate the sample standard deviation.}
    \label{fig:ala2-scale}
\end{figure}

Figure~\ref{fig:ala2-scale} shows how model accuracy is related to $N_0$ or the combination of both $N_1$ and $N_2$ for modeling the capped alanine. The scaling curve for conventional CGnet (black) is identical to the corresponding curve in Fig.~3c, i.e., the x-axis represents $N_0$. As the the flow and Flow-CGnet, the color of the curves represents the value for $N_1$, while the x-axis represents $N_2$. While the plot here displays the same information as in main text for flow models, it delivers more information for the Flow-CGnet models than the Fig.~3c, where only the effect of $N_1$ is visualized (while $N_2=7.5\times10^5$ being maximum).
We observed the following phenomena:
\begin{enumerate}
    \item For the Flow-CGnet with fixed $N_1$ (i.e., trained against the same flow model), the accuracy generally improves when $N_2$ goes up (i.e., when more flow samples are provided), except when the flow is trained with very few all-atom conformations. 
    \item For the Flow-CGnet with fixed $N_2$ (i.e., same number of CG coordinates and forces from the flow), the accuracy generally improves when $N_2$ goes up (i.e., when the flow is trained on more reference conformations). This phenomenon becomes more obvious when $N_2$ is sufficiently large.
    \item Flow-CGnet trained with the least amount of flow samples ($N_2$) from the flow trained with the least amount of reference conformations ($N_1$) has already slightly better performance than the conventional CGnet with maximum available all-atom information ($N_0$).
\end{enumerate}
Based on these, we postulate that the main advantage of flow-matching is that the flow provides ``better'' CG forces for the CGnet training. 
This could be because the flow CG forces have much less noise than the mapped all-atom forces.
In the meantime, the availability of essentially infinite flow samples helps the CGnet to better exploit the high accuracy of flow CG modeling via force-matching.

\subsection{The effect of post-processing on the sizes of flow samples}
\begin{figure}
    \centering
    \includegraphics[width=0.5\textwidth]{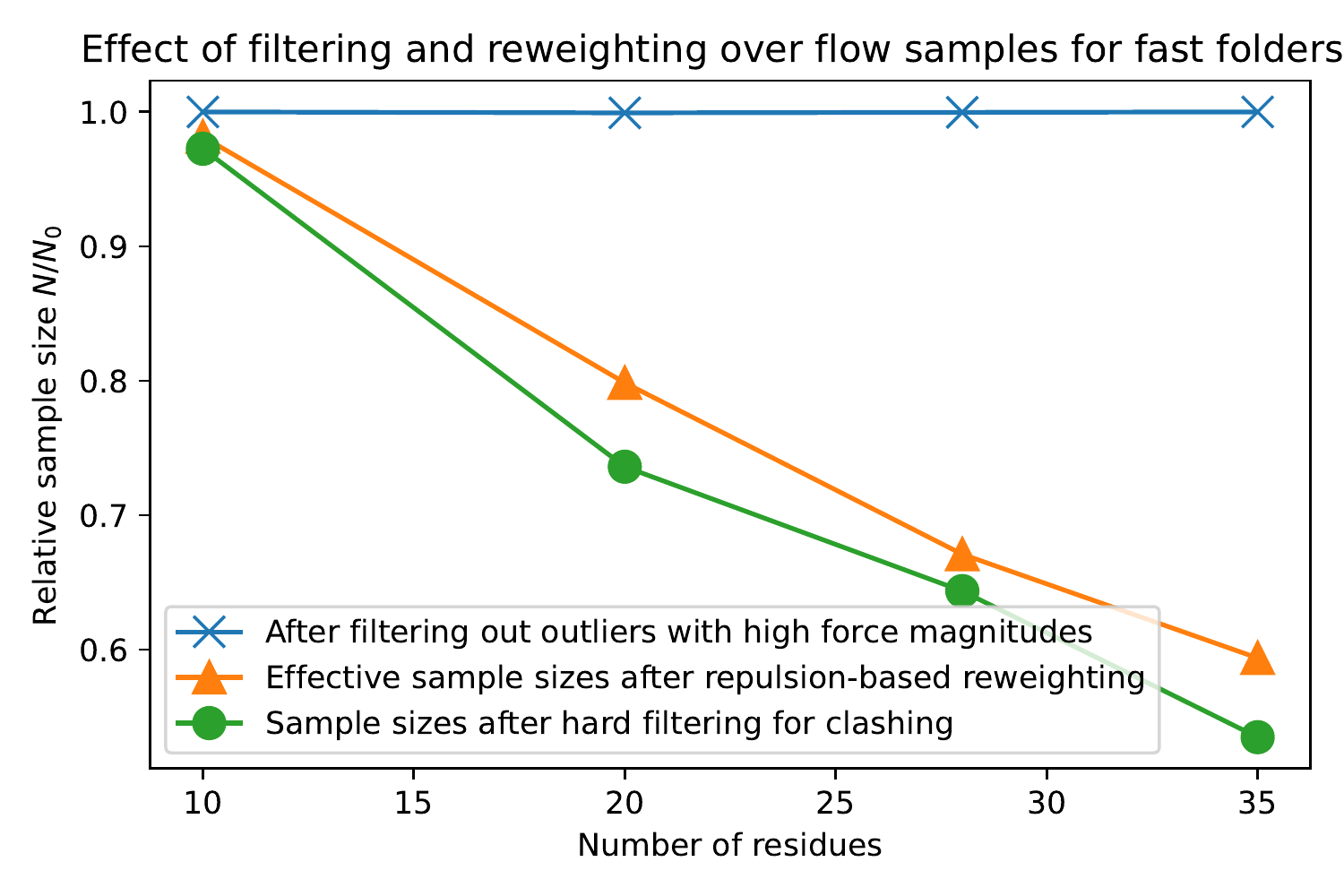}
    \caption{Fraction of remaining flow samples as a function of length of the modeled fast folding proteins after two stages of post-processing. The four columns of data points from left to right correspond to chignolin, trpcage, bba and villin, respectively. Three curves indicate the remaining sample size after different post-processing steps:~The common first step is to discard samples with high force magnitudes (resulting in the blue curve); subsequently, the samples are either reweighed according to the repulsion term (the orange curve) or simply filtered to remove samples with non-neighboring residue pairs with alpha-carbon atoms closer than 4 Angstroms.}
    \label{fig:ff-filtering}
\end{figure}

Figure~\ref{fig:ff-filtering} demonstrates the fraction of samples that have been filtered out during the post-processing step for fast folders. The x-axis shows the number of residues, i.e., the length of the modeled protein chain, while the y-axis indicates the fraction of samples remained after the processing steps. The first step, i.e., filtering out samples with extraordinarily large force magnitude only affected very few sample points, which did not scale with the chain length (Fig.~\ref{fig:ff-filtering} blue curve with cross markers). The second step, i.e., reweighing samples according to the repulsion term, affected larger fraction of samples as the chain length grew. This phenomenon is demonstrated by the orange curve with triangular markers, which plots the effective sample size according to the conventional definition \begin{equation}
    N_\text{eff}=\frac{(\sum w_i)^2}{\sum{w_i^2}}.
\end{equation}
The decrease of effective sample size can be explained by the fact that more samples tend to contain spatial clashes when the flow models a longer chain. This is indicated by the green curve with circular markers, which shows the fraction of remaining samples after an alternative secondary filtering process by simply discarding outlier samples with at least one pair of non-neighboring residues closer than 4 Angstroms.

\subsection{Comparison of representative structures}

\begin{figure}
    \centering
    \includegraphics[width=\textwidth]{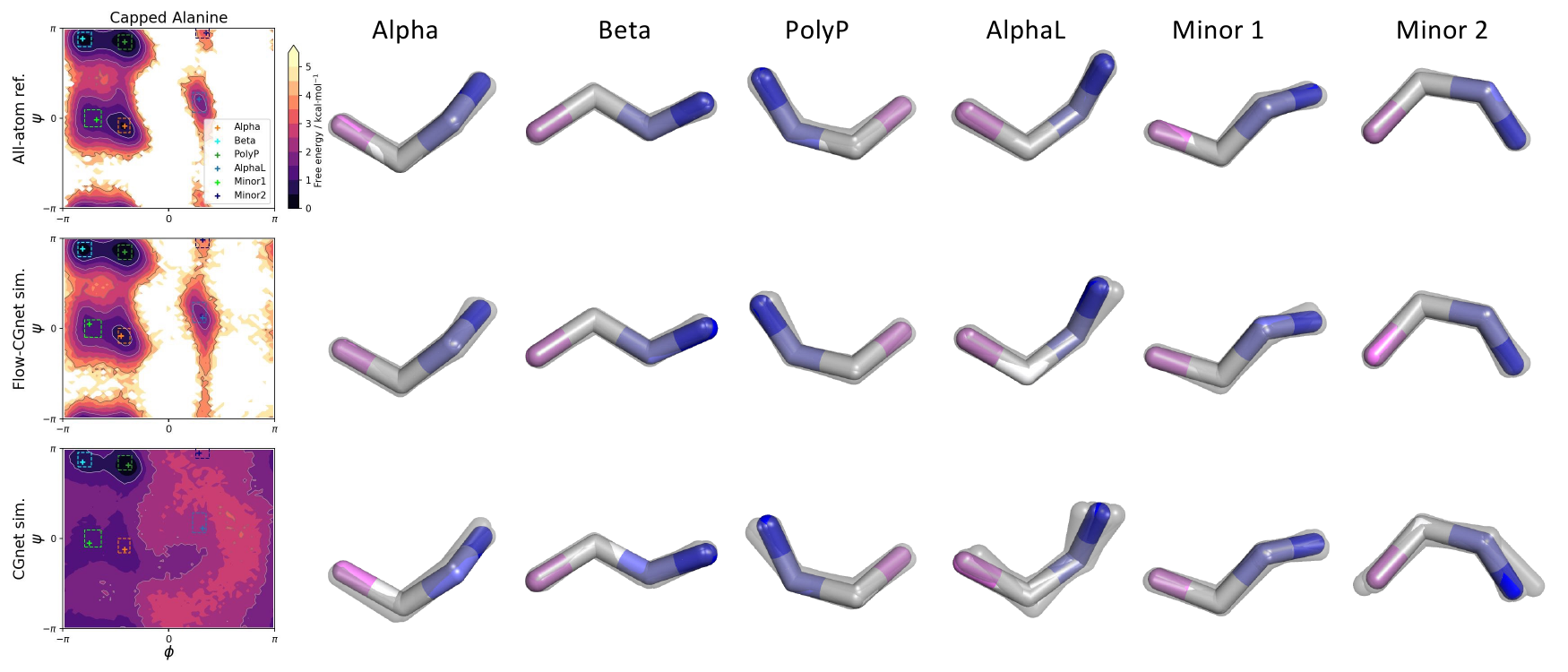}
    \caption{Comparison of structural ensembles corresponding to the different minima in the free energy landscape of capped alanine. The leftmost panel shows 2D free energy as in the main text, as well as the six regions (dashed rectangles) enclosing local free energy minima in the all-atom cases. Columns on the right display the representative structures. The three rows from top to bottom correspond to the all-atom reference, the Flow-CGnet simulation as well as the conventional CGnet simulation output. The highlighted structures in blue-white-magenta color scheme represent the conformations found at the very bottom of each basin (or simply from the bin with lowest free energy when there is no basin, e.g., for conventional CGnet), which correspond to the plus markers in the left most plot. Additionally, 10 randomly picked structures from the rectangle are shown as overlay in half-transparent gray.
    }
    \label{fig:ala-structures}
\end{figure}

Here we compare the representative structures from the energy minima observed in both all-atom reference and CG simulations. Figure~\ref{fig:ala-structures} shows the representative conformations for capped alanine from the six free energy basins on the Ramachandran plot. Since the conformations of this simple molecule are almost exclusively determined by the two dihedral angles, the structures from the same 2D region are almost identical. The Flow-CGnet outperforms the conventional CGnet mostly on the correct depth and boundary of the basins.

\begin{figure}
    \centering
    \includegraphics[width=\textwidth]{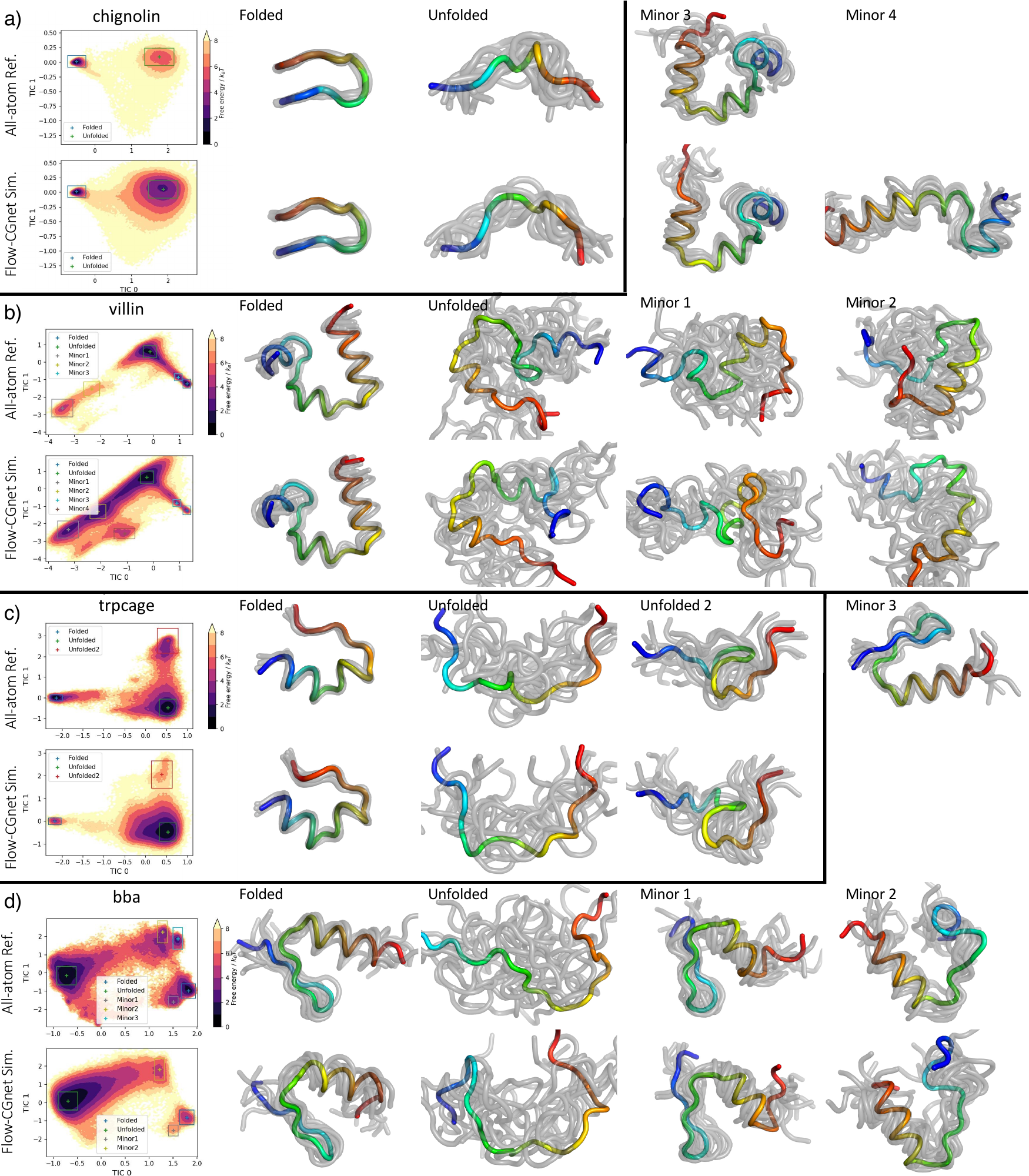}
    \caption{Comparison of structural ensembles corresponding to the different minima in the free energy landscape of the proteins considered in this study: a) chignolin; b) villin; c) trpcage; d) bba. The illustration generally follows the convention of Fig.~\ref{fig:ala-structures}, except that the locations of rectangles are not predetermined, but they are rather placed so as to keep the found free energy minima at their centers.
    }
    \label{fig:ff-structures}
\end{figure}

Figure~\ref{fig:ff-structures} illustrates the representative structures corresponding to the free energy minima for the fast-folding proteins. Unlike the plot for capped alanine, a small region on the 2D free energy landscape can map to either a well-defined (e.g., the native fold) or a drastically heterogeneous ensemble (especially for the unfolded state) of structures. For chignolin, trpcage and most minima of villin and bba, the conformational ensembles from the CG model agree well with the reference. For villin, the state labeled as "Minor 4", is only observed in the Flow-CGnet samples. Its most distinct deviation from the folded state lies in the missing kink (colored lime-yellow). For bba, the Flow-CGnet does not significantly populate the state labeled as "Minor 3". This is a misfolded state in the all-atom MD simulation, featuring a register shift in the beta-strand from the native fold. The difference consists of the location of the beta-turn as well as the relative position of the N-tail (blue) and the beta-alpha junction (green). The reason for the missing structural element in villin as well as the alternative beta-strand in bba could be an underestimated relative stability of these states in the Flow-CG, which is observed also between the folded and unfolded states in general.

\providecommand{\latin}[1]{#1}
\makeatletter
\providecommand{\doi}
  {\begingroup\let\do\@makeother\dospecials
  \catcode`\{=1 \catcode`\}=2 \doi@aux}
\providecommand{\doi@aux}[1]{\endgroup\texttt{#1}}
\makeatother
\providecommand*\mcitethebibliography{\thebibliography}
\csname @ifundefined\endcsname{endmcitethebibliography}
  {\let\endmcitethebibliography\endthebibliography}{}